\documentclass[12pt]{article}
\usepackage{eurosym}
\usepackage{amssymb}
\usepackage{amsmath}
\usepackage{amsfonts}
\usepackage{hyperref}
\usepackage{theorem}
\usepackage{graphicx}
\usepackage{dsfont}
\usepackage{color}
\usepackage{here}
\usepackage{mathtools}

\newtheorem{theorem}{Theorem}

\newtheorem{corollary}{Corollary}

\newtheorem{definition}{Definition}

\newtheorem{lemma}{Lemma}

\newtheorem{proposition}{Proposition}

\newenvironment{proof}[1][Proof]{\noindent\textbf{#1.} }{\ \rule{0.5em}{0.5em}}
\newcommand{\localinc}[3]{\overset{#1,#2 ,#3}{\hookrightarrow }}
\newcommand{\localincstep}[2]{\overset{#1,#2}{\longrightarrow }}

\begin{document}

\title{An Axiomatization of the Shapley-Shubik 
Index for Interval Decisions}
\author{}
\maketitle

\author{Sascha Kurz\footnote{Department of Mathematics, Physics and
Computer Science, University of Bayreuth, 95440 Bayreuth, Germany
Tel.: +49-921-557353   Fax: +49-921-557352.  E--mail: sascha.kurz@uni-bayreuth.de}
\  {\normalsize and} \ 
Issofa Moyouwou\footnote{
Advanced Teachers Training College, University of Yaounde I, PO Box 47 Yaounde, Cameroon}, 
and Hilaire Touyem\footnote{Research and Training Unit for Doctorate in Mathematics, Computer Sciences and Applications, University of Yaounde I, PO Box 812, Yaounde, Cameroon}

\begin{abstract}
  The Shapley-Shubik index was designed to evaluate the power distribution in committee systems 
  drawing binary decisions and is one of the most established power indices. It was generalized 
  to decisions with more than two levels of approval in the input and output. In the limit we have 
  a continuum of options. For these games with interval decisions we prove an axiomatization of a 
  power measure and show that the Shapley-Shubik index for simple games, as well as for $(j,k)$~simple 
  games, occurs as a special discretization. This relation and the closeness of the stated axiomatization 
  to the classical case suggests to speak of the Shapley-Shubik index for games with interval decisions,  
  that can also be generalized to a value.
  
  \vskip 3mm

  \noindent {\it Math. Subj. Class. (2000)}: Primary 
  91A40, 91A80, 91B12.

  \noindent {\it JEL Class.}: C71, D70, D71.
\end{abstract}

\section{Introduction}
\label{sec_introduction}

Consider a committee, which jointly draws decisions by voting and aggregating the votes to a common outcome. The case that is 
studied mostly in the literature is that of binary decisions, i.e., for a given proposal each committee member can either 
vote {\lq\lq}yes{\rq\rq} or {\lq\lq}no{\rq\rq}. The aggregated group decision then also is either to accept the proposal and 
implement it or to dismiss the proposal and retain the status quo. In the case of heterogeneous committee members the question of 
their respective influence on the group decision arises. Taking just the formal specification of the given voting rule into 
account, a measurement for influence is called a power index. A rather comprehensive and widely applied class of voting rules is 
that of simple games. One of the most established power indices for simple games is the Shapley-Shubik index \cite{shapley1954method}.

However, not all decisions are binary. Abstaining from a vote might be seen as a third option for the committee members. In general, 
there might also be any number $j\ge 2$ of alternatives that can be chosen from. To this end, simple games were generalized to 
$(j,k)$ simple games \cite{freixas2003weighted}, where $j$ is the number of alternatives in the input, i.e., the voting possibilities, 
and $k$ the number of alternatives for the group decision. A Shapley-Shubik power index for $(3,2)$ simple games was 
introduced in \cite[pp. 291--293]{felsenthal1998measurement}. When discussing the so-called roll call model for the Shapley-Shubik 
index, we will see that certain biases of the voters to {\lq\lq}yes{\rq\rq} or {\lq\lq}no{\rq\rq}-votes do not matter for 
the Shapley-Shubik index for simple games. This changes if voters have at least a third option. So given some probability that voters do 
not abstain, Friedman and Parker consider a conditional Shapley-Shubik power index for $(3,2)$ simple games \cite{friedman2018conditional}.    
For general $(j,k)$ simple games a Shapley-Shubik power index was introduced in \cite{freixas2005shapley}. For a 
variant with a different notion of monotonicity see \cite{hsiao1993shapley}. 
For completeness, we mention that \cite{amer1998extension} considers a Shapley value for $r$ unordered alternatives in the input. 
Axiomatizations are still under the way and mainly consider special cases.

In some applications, i.e., tax rates, we have a continuum of options, which might be modeled by a real number. The aggregation of 
$n$ different real numbers to a single real number was studied in the literature under the name of aggregation functions, see e.g.\ 
\cite{grabischaggregation}. Mimicking the properties of a simple game we speak of interval simple games for real-valued decisions in 
$[0,1]$. A generalization of the Shapley-Shubik index to that context was proposed in \cite{kurz2014measuring}. 
Here we give an axiomatic justification for that power index. Moreover, we show that simple games as well as $(j,k)$ simple games 
are covered by interval simple games using a natural embedding. This gives a coherent story condensing the different variants 
for committee decisions in one common framework.

As further related work one might mention considerations for other power indices. An axiomatization of a Banzhaf index for $(j,k)$ simple games
was obtained in \cite{freixas2005banzhaf}. For games with abstention, $(3,2)$ simple games, the ordinal equivalence between the Shapley-Shubik and the  
Banzhaf index has been studied in \cite{tchantcho2008voters}. Further influence measures have been studied in \cite{grabisch2010model}, see 
also \cite{freixas2012prob}. Models of influence for a continuum of actions have e.g.\ been studied in \cite{abdou1988neutral,grabisch2011model}. 
For aggregation functions there is a known generalization of the (Penrose-)Banzhaf index \cite{grabischaggregation}, see also \cite{kurz2018importance} 
for more details on the relationship.


The remaining part of this paper is structured as follows. In Section~\ref{sec_preliminaries} we introduce 
the basic notions, i.e., we consider committee decisions in Subsection~\ref{subsec_committee_decisions},   
power indices in Subsection~\ref{subsec_power_indices}, and the roll call model in Subsection~\ref{subsec_roll_call}. 
The generalization of the Shapley-Shubik index to interval simple games is studied in Section~\ref{sec_the_index}.
The main part, i.e., an axiomatization of that index is given in Section~\ref{sec_axiomatization}. After studying the first basic properties 
and observing that the classical axioms are not sufficient to uniquely characterize the index, we introduce the 
new axiom (HIS) of homogeneous increments sharing in Subsection~\ref{subsec_HIS}. In terms of simple games this axiom 
corresponds to the axiom (SymGL) of symmetric gain-loss \cite[p. 93]{laruelle2001shapley}. Technically we will work with step functions, 
which can approximate any reasonable general function. Assuming that a power index is commutable with the limit of 
step functions, we end up with an axiomatization, see Theorem~\ref{thm_characterization} in Subsection~\ref{subsec_characterization}. In 
order to illustrate the technical details and subtleties we give a detailed example in the appendix. In Section~\ref{sec_conclusion} we 
close the paper with a conclusion.

\section{Preliminaries}
\label{sec_preliminaries}
Let $N=\left\{ 1,2,...,n\right\}$ be  a finite set of voters. Any non-empty subset $S$ of $N$ is 
called a coalition and the set of all coalitions of $N$ is denoted by $2^{N}$. For easier reading capital
letters are reserved for coalitions (such as $N$, $S$, $T$, $J$, $K$, \dots), while the corresponding small letters 
($n$, $s$, $t$, $j$, $k$, \dots) denote their respective cardinalities. By $\mathcal{S}_n$ we denote the set of 
permutations of length $n$, i.e., the bijections on $N$. 

\subsection{Committee decisions}
\label{subsec_committee_decisions}
A most easy framework for committee decisions are binary decisions, i.e., each committee member can either vote 
{\lq\lq}no{\rq\rq} or {\lq\lq}yes{\rq\rq} on a given proposal, while the group decision then is either to reject 
or to accept the proposal. This setting is commonly formalized by a \emph{simple game} $v$ (on $N$), which is 
a mapping $v\colon 2^{N}\to \left\{ 0,1\right\}$ such that $v(\emptyset) =0$, $v(N)=1$, and $v(S)\leq v(T)$ for
all coalitions $\emptyset\subseteq S\subseteq T\subseteq N$. Here $S\subseteq N$ collects all players that are 
voting {\lq\lq}yes{\rq\rq} and we have $v(S)=1$ iff coalition $S$ can bring trough the proposal. If $v(S)=1$, then
coalition $S$ is called \emph{winning} and \emph{losing} otherwise. By $[q;w_1,\dots,w_n]$ we denote the simple 
game whose winning coalitions are exactly those with $w(S):=\sum_{i\in S} w_i\ge q$. In this case we also speak 
of a \emph{weighted game} (with weights $w_1,\dots,w_n$ and quota $q$). The term \textit{simple} 
refers to the fact that there are just two options in the input as well as in the output. In some older literature 
a simple game does not need to satisfy the monotonicity assumption $v(S)\leq v(T)$ and one speaks of \emph{monotonic 
simple games} if it does. However, the monotonicity assumption is quite natural for most decisions, since it only 
requires that additional supporters for a proposal do not turn acceptance into rejection. We follow the more recent literature 
and just speak of simple games.

In several applications decisions are not binary. One may think of grades or situations where abstention is possible.
To this end, in \cite{freixas2003weighted}, see also \cite{freixas2009anonymous}, the authors have introduced 
\emph{$(j,k)$~simple games}, which map the selection of $n$ voters for $j$ possible levels of approval to 
$k$ possible outputs. We slightly reformulate their definition. For $J=\{0,1,\dots, j-1\}$ and $K=\{0,1,\dots,k-1\}$, 
where $j,k\ge 2$, a $(j,k)$~simple game $v$ (on $N$) is a mapping $v\colon J^n\to K$ with $v(0,\dots,0)=0$, $
v(j-1,\dots,j-1)=k-1$, and $v(x)\le v(y)$ for all $x,y\in J^n$ with $x\le y$. Here we write $x\le y$ for $x,y\in\mathbb{R}^n$ if 
$x_i\le y_i$ for all $1\le i\le n$. A $(2,2)$-simple game is isomorphic to a simple game. Note that the input levels from 
$J$, as well as the output levels from $K$, are assumed to be ordered to make the monotonicity condition meaningful. For 
$r$ unordered input alternatives, see e.g.\ \cite{bolger1986power}.   

If we rescale the input and output levels of a $(j,k)$~simple game to $\frac{1}{j-1}\cdot(0,\dots,j-1)$ and 
$\frac{1}{k-1}\cdot(0,\dots,k-1)$,\footnote{Technically we will use a slightly different scaling, see Footnote~\ref{fn_rescaling}.} respectively, 
then the input and output levels are both contained in the real 
interval $[0,1]$ between $0$ and $1$. Increasing $j$ and $k$ approximates $[0,1]$, so that we may consider its limit 
$[0,1]$ itself. So, we want to study $[0,1]^n\to[0,1]$-functions $v$ with $v(\mathbf{0})=0$, $v(\mathbf{1})=1$, and $v(x)\le v(y)$ for 
all $x,y\in [0,1]^n$ with $x\le y$. In \cite{kurz2014measuring} the author called those objects continuous simple games 
since, for simplicity, $v$ was assumed to be continuous. To go in line with the above naming we call them 
\emph{interval simple games} here\footnote{In \cite{kurz2018importance} they were called simple aggregation functions, so that 
in any case the naming should be considered as temporary.}. Without the monotonicity assumption $v(x)\le v(y)$ and the domain restriction 
to $[0,1]$ those functions are more widely known under the name \emph{aggregation function}, see e.g.\ \cite{grabischaggregation}. 
The name says it, an aggregation function takes $n$ real numbers as inputs and condenses them to a single real number.       
Examples are direct votes on e.g., top tax rates or pension contributions, i.e., real numbers from some interval can be 
directly named instead of approving or disapproving some concrete proposal. If the real-valued alternatives are ordered 
by their usual order, then the monotonicity assumption makes sense again. In the voting context, the weighted median 
is a reasonable aggregation function, see e.g.\ \cite{kurz2017democratic} for the assignment of \textit{fair} weights 
in the corresponding two-tier context. Mathematically, also a function like e.g.\ $v(x_1,\dots,x_n)=\prod_{i=1}^n x_i^i$ 
falls into the class of interval simple games. With respect to the restriction to the specific interval $[0,1]$ we note that 
this can be achieved by rescaling, so that we retain this here due to simplicity. However, higher dimensional policy spaces 
are significantly different from our setting.

\subsection{Power indices}
\label{subsec_power_indices}
Even if the case where a simple game $v$ is weighted, 
influence or power is not always reasonably reflected by 
the weights. This fact is well-known and triggered the invention of power indices, i.e., mappings from a simple game on 
$n$ players to $\mathbb{R}^n$ reflecting the influence of a player on the final group decision. One of the most established power indices 
is the \emph{Shapley-Shubik index} \cite{shapley1954method}. It can be defined via
\begin{equation}
  \label{eq_ssi_simple_games}
  \operatorname{SSI}_{i}(v)=\sum_{i\in S\subseteq N}\frac{(s-1)!(n-s)!}{n!}\cdot \left[
  v(S)-v(S\backslash \left\{ i\right\} )\right]
\end{equation}
for all players $i\in N$. If $v(S)-v(S\backslash\{i\})=1$, then we have $v(S)=1$ and $v(S\backslash\{i\})=0$ in a simple game 
and voter~$i$ is called a swing voter. In the next subsection we give another equivalent formulation for $\operatorname{SSI}_{i}(v)$ 
based on the so-called roll call model, where pivotality plays the essential role and nicely motivates the factors $\frac{(s-1)!(n-s)!}{n!}$. 
This interpretation triggers the definition of a Shapley-Shubik index for $(j,k)$~simple games in \cite{freixas2005shapley} and was 
generalized to interval simple games in \cite{kurz2014measuring}.  

Another way to characterize power indices are axiomatizations, i.e., sets of properties that are satisfied by a power index and uniquely 
characterize it. For the Shapley-Shubik index we refer to \cite{0050.14404,shapley1954method}. In order to introduce properties of power 
indices for all three types of games, let $v$ be a mapping $J^n\to K$, where $J=\{0,1,\dots,j-1\}$ for some integer $j\ge 2$ or $J=[0,1]$, 
and $K=\{0,1,\dots,k-1\}$ for some integer $k\ge 2$ or $K=[0,1]$. A \emph{power index} $\varphi$ maps $v$ to $\mathbb{R}^n$ for all 
$v\in\mathcal{V}^n_{J,K}$, where $\mathcal{V}^n_{J,K}$ denotes the set of all corresponding games for $n$ players. We call $\varphi$ 
\emph{positive} if $\varphi(v)\neq 0$ and $\varphi_i(v)\ge 0$ for all $v\in\mathcal{V}^n_{J,K}$ and all $i\in N$. If $\sum_{i=1}^n \varphi_i(v)=1$
for all games $v$, then $\varphi$ is called $\emph{efficient}$. A power index $\varphi$ is called \emph{anonymous} if we 
have $\varphi_{\pi(i)}(\pi v)=\varphi_i(v)$ for all permutations $\pi$ of $N$, $i\in N$, and $v\in\mathcal{V}_{J,K}^n$, where 
$\pi v (x)=v(\pi(x))$ and $\pi(x)=\left(x_{\pi(i)}\right)_{i\in N}$. If $\pi\in\mathcal{S}_n$ is a transposition interchanging 
player~$i$ and player~$j$, then we call the two players \emph{symmetric} if $\pi v (x)=v(x)$ for all $x\in J^n$. A power index $\varphi$ 
is called \emph{symmetric} if $\varphi_i(v)=\varphi_j(v)$ for all players $i$ and $j$ that are symmetric in $v$. Note that symmetry 
is a relaxation of anonymity. A player $i\in N$ is called a \emph{null player} if 
$v(x)=v(y)$ for all $x,y\in J$ with $x_j=y_j$ for all $j\in N\backslash \{i\}$, i.e., $v(x)$ does not depend on $x_i$. A power index $\varphi$ 
is said to \emph{satisfy the null player property} if $\varphi_i(v)=0$ for every null player $i$ in $v$. In general, a function 
$f\colon\mathbb{R}^n\supseteq U\to\mathbb{R}$ is called \emph{linear} if we have $f(\alpha x+\beta y)=\alpha f(x)+\beta f(y)$ for 
all $x,y\in U$ and $\alpha,\beta\in\mathbb{R}$ such that $\alpha x+\beta y\in U$. For simple games (and their generalizations) linear combinations 
$\alpha u+\beta v$ of simple games $u$ and $v$ are almost never a simple game again, so that linearity has been adopted to the following. 
A power index $\varphi$ satisfies the \emph{transfer property} if for all $u,v\in \mathcal{V}_{J,K}^n$ and all $i\in N$ we have 
$\varphi_i(u)+\varphi_i(v)=\varphi_i(u\vee v)+\varphi_i(u\wedge v)$, where $(u\vee v)(x)=\max\{u(x),v(x)\}$ and $(u\wedge v)(x)=\min\{u(x),v(x)\}$ 
for all $x\in J^n$. Note that we always have $u\vee v,u\wedge v\in\mathcal{V}_{J,K}^n$. With this we can state that the Shapley-Shubik index 
for simple games is the unique power index that is symmetric, efficient, satisfies both the null player property and the transfer property, 
see \cite{dubey1975uniqueness}. Moreover, $\operatorname{SSI}$ is also positive and anonymous. The more general Shapley value, coinciding 
with the Shapley-Shubik index for simple games and having the set of all cooperative games for $n$ players as domain, is linear.  

\subsection{The roll call model}
\label{subsec_roll_call}
In \cite{shapley1954method} the authors have motivated the Shapley-Shubik index by the following interpretation. Assume that 
the $n$ voters row up in a line and declare to be part in the coalition of {\lq\lq}yes{\rq\rq}-voters. Given an ordering of the players, 
the player that first guarantees that a proposal can be put through is then called pivotal. Considering all $n!$ orderings $\pi\in\mathcal{S}_n$ of 
the players with equal probability then gives a probability for being pivotal for a given player $i\in N$ that equals its 
Shapley-Shubik index. So we can rewrite Equation~(\ref{eq_ssi_simple_games}) to 
\begin{equation}
  \label{eq_roll_call_simple_games_first}
  \operatorname{SSI}_i(v)=\frac{1}{n!}\cdot \sum_{\pi\in\mathcal{S}_n} \Big( v(\{j\in N\,:\,\pi(j)\le \pi(i)\})-v(\{j\in N\,:\,\pi(j)< \pi(i)\}) \Big). 
\end{equation}
Setting $S_\pi^i:=\{j\in N\,:\,\pi(j)\le \pi(i)\}$ we have $S_\pi^i=S$ for exactly $(s-1)!(n-s)!$ permutations $\pi\in\mathcal{S}_n$ and 
an arbitrary set $\{i\}\subseteq S\subseteq N$, so that Equation~(\ref{eq_ssi_simple_games}) is just a simplification of 
Equation~(\ref{eq_roll_call_simple_games_first}).

Instead of assuming that all players vote {\lq\lq}yes{\rq\rq} one can also assume that all players vote 
{\lq\lq}no{\rq\rq}. In \cite{MannShapley} it is mentioned that the model also yields the same result if we assume 
that all players independently vote {\lq\lq}yes{\rq\rq} with a fixed probability $p\in[0,1]$. This was further generalized to probability 
measures $p$ on $\{0,1\}^n$ where vote vectors with the same number of {\lq\lq}yes{\rq\rq} votes have the same probability, see 
\cite{hu2006asymmetric}. In other words, individual votes may be interdependent but must be exchangeable. That no further probability 
measures lead to the Shapley-Shubik index was finally shown in \cite{kurz2018roll}. For the most symmetric case $p=\tfrac{1}{2}$ we can rewrite 
Equation~(\ref{eq_roll_call_simple_games_first}) to 
\begin{equation}
  \label{eq_roll_call_simple_games}
  \operatorname{SSI}_i(v)=\frac{1}{n!\cdot 2^n}\cdot \sum_{(\pi,x)\in\mathcal{S}_n\times\{0,1\}^n} M(v,(\pi,x),i), 
\end{equation}
where $M(v,(\pi,x),i)$ is one if player $i$ is pivotal for ordering $\pi$ and vote vector $x$ in $v$, see \cite{kurz2018roll}, 
and zero otherwise.

Being pivotal means that the vote of player~$i$, according to the ordering $\pi$ and the votes of the previous players, fixes 
the group decision for the first time. In a $(j,k)$~simple game we can have the same notation as long as there are just $k=2$ 
outputs. If $k>2$, then pushing the outcome to at least $h$ or at most $h-1$ are possible events for $h\in\{1,\dots,k-1\}$, so that we 
speak of an $h$-pivotal player, which is always unique, see \cite{freixas2005shapley}. In that paper the author defines 
\begin{equation}
  \label{eq_roll_call_jk_simple_games}
  \frac{1}{n!\cdot j^n\cdot (k-1)}\sum_{h=1}^{k-1} \left|\left\{(\pi,x)\in\mathcal{S}_n\times J^n\,:\,i\text{ is an $h$-pivot for 
  $\pi$ and $x$ in $v$}\right\}\right|,
\end{equation}
for all $i\in N$, as the Shapley-Shubik index for $(j,k)$~simple games. Using monotonicity the $h$-pivotality as well as $M(v,(\pi,x),i)$ 
can be stated more directly. Slightly abusing notation we write $\mathbf{0}\in\mathbb{R}^n$ and $\mathbf{1}\in\mathbb{R}^n$ for the 
vectors that entirely consist of zeroes and entries $j-1$, respectively. For each $\emptyset\subseteq S\subseteq N$ we write $x_S$ for the restriction 
of $x\in\mathbb{R}^n$ to $\left(x_i\right)_{i\in S}$. As an abbreviation, we write $x_{-S}=x_{N\backslash S}$. For a given permutation 
$\pi\in\mathcal{S}_n$ and $i\in N$, we set $\pi_{<i}=
\left\{j\in N\,:\, \pi(j)<\pi(i)\right\}$, $\pi_{\le i}= \left\{j\in N\,:\, \pi(j)\le\pi(i)\right\}$, $\pi_{>i}=\left\{j\in N\,:\, 
\pi(j)>\pi(i)\right\}$, and $\pi_{\ge i}= \left\{j\in N\,:\, \pi(j)\ge \pi(i)\right\}$. With this, we can rewrite 
(\ref{eq_roll_call_jk_simple_games}) to 
\begin{equation}
  \label{eq_uncertainty_reduction_jk_simple_games}
  \frac{1}{n!\!\cdot\! j^n\!\cdot\! (k-1)}\sum_{(\pi,x)\in\mathcal{S}_n\times J^n}\!\! \Big(\!\!\left[v(x_{\pi_{<i}},\mathbf{1}_{\pi_{\ge i}})-v(x_{\pi_{<i}},\mathbf{0}_{\pi_{\ge i}})\right]
  -\left[v(x_{\pi_{\le i}},\mathbf{1}_{\pi_{> i}})-v(x_{\pi_{\le i}},\mathbf{0}_{\pi_{> i}})\right]\!\!\Big).
\end{equation}  
The interpretation is as follows. Since $v$ is monotone, before the vote of player~$i$ exactly the values in 
$\left\{v(x_{\pi_{<i}},\mathbf{0}_{\pi_{\ge i}}),\dots, v(x_{\pi_{<i}},\mathbf{1}_{\pi_{\ge i}})\right\}$ are still possible as final group decision. After 
the vote of player~$i$ this interval eventually shrinks to $\left\{v(x_{\pi_{\le i}},\mathbf{0}_{\pi_{> i}}),\dots, v(x_{\pi_{\le i}},\mathbf{1}_{\pi_{> i}})\right\}$. 
The difference in (\ref{eq_uncertainty_reduction_jk_simple_games}) just computes the difference between the lengths of both intervals, i.e., 
the number of previously possible outputs that can be excluded for sure after the vote of player~$i$.

In order to simplify~(\ref{eq_uncertainty_reduction_jk_simple_games}) a bit, let $C(v,T)=\frac{1}{j^n(k-1)}\cdot \sum_{x\in J^n} \Big(v(\mathbf{1}_T,x_{-T})-v(\mathbf{0}_T,x_{-T})\Big)$ for all $T\subseteq N$. 
As in the situation where we simplified the Shapley-Shubik index of a simple game given by Equation~(\ref{eq_roll_call_simple_games_first}) 
to Equation~(\ref{eq_ssi_simple_games}), we observe that it is sufficient to know the sets $\pi_{\ge i}$ and $\pi_{>i}$ for every permutation 
$\pi\in\mathcal{S}_n$. So we can condense all permutations that lead to the same set and can simplify (\ref{eq_uncertainty_reduction_jk_simple_games}) 
to  
\begin{equation}
  \label{eq_roll_call_jk_simple_games_simplified}
  \sum_{i\in S\subseteq N} \frac{(s-1)!(n-s)!}{n!}\cdot\left[C(v,S)-C(v,S\backslash\{i\})\right].
\end{equation}
Note the similarity between (\ref{eq_roll_call_jk_simple_games_simplified}) and Equation~(\ref{eq_ssi_simple_games}).  
For $j=k=2$, i.e., simple games, the coincidence between $v(S)-v(S\backslash \left\{ i\right\})$ and the more complicated 
summation behind $C(v,S)-C(v,S\backslash\{i\})$ is due to the fact that the roll call model gives the same probabilities for 
$p=1$ and $p=\tfrac{1}{2}$.\footnote{A direct combinatorial proof of the underlying identity was also given in \cite{bernardi2018shapley,
kurz2016generalized}.} However, this is an artifact for $j=2$ and for $j>2$ different probabilities for the input levels, 
as well as more complicated probability distributions on vote vectors, lead to different results in the roll call model. The 
case $(j,k)=(3,2)$ was studied in more detail in \cite{friedman2018conditional}, where the authors defined a conditional Shapley-Shubik index 
given some fixed probability of the voters to abstain.

\section{A Shapley-Shubik like index for interval decisions}
\label{sec_the_index}

If we renormalize $J$ and $K$ to subsets of $[0,1]$ and consider the limit taking $j$ and $k$ to infinity,\footnote{To be more precisely, we consider 
the mappings $K=\{0,1,\dots,k-1\}\to[0,1]$, $i\mapsto i/(k-1)$ and $J=\{0,1,\dots,j-1\}\to[0,1]$, $i\mapsto \left(i+\tfrac{1}{2}\right)/j$. Note 
that the latter points are the middle points of the intervals $(i/j,(i+1)/j)$ for $0\le i\le j-1$, c.f.\ Footnote~\ref{fn_center} 
and Definition~\ref{def_natural_embedding}.\label{fn_rescaling}} 
then Equation~(\ref{eq_uncertainty_reduction_jk_simple_games}) gives a power index for interval simple games. 

\begin{definition}(cf.~\cite[Definition 6.2]{kurz2014measuring})\\
  Let $v$ be an interval simple game with player set $N$ and $i\in N$ an arbitrary player. We set 
  \begin{eqnarray} 
  \Psi_i(v) &=& \frac{1}{n!}\sum_{\pi\in\mathcal{S}_n}\int_0^1\dots\int_0^1 \left[v(x_{\pi_{<i}},\mathbf{1}_{\pi_{\ge i}})-v(x_{\pi_{<i}},\mathbf{0}_{\pi_{\ge i}})\right]\notag\\
  &&-\left[v(x_{\pi_{\le i}},\mathbf{1}_{\pi_{> i}})-v(x_{\pi_{\le i}},\mathbf{0}_{\pi_{> i}})\right]\operatorname{d}x_1\dots \operatorname{d}x_n.\label{eq_def_ssi_interval_sg}
  \end{eqnarray}\label{def_ssi_interval_sg}
\end{definition}

We remark that we change between the notations $\int_0^1\dots\int_0^1$ and $\int_{[0,1]^n}$ from time to time taking Fubini's 
theorem into account. We skip the question for the existence of the involved integrals till we state a simpler formula for $\Psi$ 
in Proposition~\ref{prop_formula_ssi_interval_sg}. However, a few explicit formulas have been obtained for special classes of interval 
simple games directly using the rather complicated expression from Definition~\ref{def_ssi_interval_sg}.

\begin{proposition}(\cite[Theorem 6.3]{kurz2018importance})\\
Let $w\in\mathbb{R}_{\ge 0}^n$ with $\sum_{i=1}^n w_i=1$ and $f_i\colon[0,1]\to[0,1]$ weakly monotonic increasing functions with $f(\mathbf{0})=0$ and 
$f(\mathbf{1})=1$ for all $i\in N$. Then $f\colon[0,1]^n\to[0,1]$ defined by $x\mapsto \sum_{i=1}^n w_i\cdot f_i(x_i)$ is an interval simple game 
and satisfies $\Psi_i(f)=w_i$ for all $i\in N$.
\end{proposition}

\begin{proposition}(\cite[Theorem 6.4]{kurz2018importance})\\ \label{prop_formula_exponential_product}
  For a positive integer $n$ and positive real numbers $\alpha_1,\dots,\alpha_n$ let $f\colon[0,1]^n\to[0,1]$ be defined by 
  $x\mapsto \prod_{i=1}^n x_i^{\alpha_i}$ and $\Lambda=\prod_{j=1}^n \left(\alpha_j+1\right)$. Then, $f$ is an interval simple game and
  \begin{equation}
    \Psi_i(f)=\frac{1}{n!\cdot\Lambda}\cdot\left( (n-1)!+\alpha_i\cdot \sum_{T\subseteq N\backslash\{i\}} |T|!\cdot(n-1-|T|)! \cdot 
    \prod_{j\in T} \left(\alpha_j+1\right)\right) 
  \end{equation}
  for all $i\in N$.
\end{proposition}

Of course, it would be interesting to compute $\Psi$ for other parametric classes of interval simple games. The subsequent simplification 
in Proposition~\ref{prop_formula_ssi_interval_sg} might be rather useful for that aim. Prior to stating and proving a simplification 
of Definition~\ref{def_ssi_interval_sg}, we show that $\Psi$ indeed shares some properties that we might expect from a meaningful 
power index. 

\begin{proposition}(cf.~\cite[Definition 6.7]{kurz2014measuring}, \cite[Lemma 6.1]{kurz2018importance})\\
  The mapping $\Psi$ is positive, efficient, anonymous, symmetric, and satisfies both the null player and the transfer property.\label{proposition_power_properties} 
\end{proposition}
\newcommand{\upim}{x_{\pi_{<i}},\mathbf{1}_{\pi_{\ge i}}}
\newcommand{\upi}{x_{\pi_{\le i}},\mathbf{1}_{\pi_{>i}}}
\newcommand{\dnim}{x_{\pi_{<i}},\mathbf{0}_{\pi_{\ge i}}}
\newcommand{\dni}{x_{\pi_{\le i}},\mathbf{0}_{\pi_{>i}}}
\newcommand{\upimk}{x_{\kappa_{<i}},\mathbf{1}_{\kappa_{\ge i}}}
\newcommand{\upik}{x_{\kappa_{\le i}},\mathbf{1}_{\kappa_{>i}}}
\newcommand{\dnimk}{x_{\kappa_{<i}},\mathbf{0}_{\kappa_{\ge i}}}
\newcommand{\dnik}{x_{\kappa_{\le i}},\mathbf{0}_{\kappa_{>i}}}
\newcommand{\upimky}{y_{\kappa_{<j}},\mathbf{1}_{\kappa_{\ge j}}}
\newcommand{\upiky}{y_{\kappa_{\le j}},\mathbf{1}_{\kappa_{>j}}}
\newcommand{\dnimky}{y_{\kappa_{<j}},\mathbf{0}_{\kappa_{\ge j}}}
\newcommand{\dniky}{y_{\kappa_{\le j}},\mathbf{0}_{\kappa_{>j}}}
\newcommand{\upimkx}{x_{\kappa_{<j}},\mathbf{1}_{\kappa_{\ge j}}}
\newcommand{\upikx}{x_{\kappa_{\le j}},\mathbf{1}_{\kappa_{>j}}}
\newcommand{\dnimkx}{x_{\kappa_{<j}},\mathbf{0}_{\kappa_{\ge j}}}
\newcommand{\dnikx}{x_{\kappa_{\le j}},\mathbf{0}_{\kappa_{>j}}}
\newcommand{\upjm}{x_{\pi_{<j}},\mathbf{1}_{\pi_{\ge j}}}
\newcommand{\upj}{x_{\pi_{\le j}},\mathbf{1}_{\pi_{>j}}}
\newcommand{\dnjm}{x_{\pi_{<j}},\mathbf{0}_{\pi_{\ge j}}}
\newcommand{\dnj}{x_{\pi_{\le j}},\mathbf{0}_{\pi_{>j}}}
\begin{proof}
  For each $x\in[0,1]^n$, $\pi\in\mathcal{S}_n$, and $i\in N$, we have $f(\upim)\ge f(\upi)$ 
  and $f(\dni)\ge f(\dnim)$, so that $\Psi_i(v)\ge 0$. Since we will show that $\Psi$ is efficient, we especially 
  have $\Psi(v)\neq \mathbf{0}$, so that $\Psi$ is positive.
  
  For any permutation $\pi\in\mathcal{S}_n$ and any $0\le h\le n$ let $\pi|h:=\{\pi(i)\,:\,1\le i\le h\}$, i.e., 
  the first $h$ agents in ordering $\pi$. Then, for any state vector $x\in[0,1]^n$, we have
  \begin{eqnarray*}
    &&\sum_{i=1}^n \Big(v(\upim)-v(\upi)+v(\dni)-v(\dnim)\Big)\nonumber\\
    &=&\sum_{h=1}^n \Big(v(x_{\pi|h-1},\mathbf{1}_{-\pi|h-1})\!-\! v(x_{\pi|h},\mathbf{1}_{-\pi|h})\Big)
    \!+\!\sum_{h=1}^n \Big(v(x_{\pi|h},\mathbf{0}_{-\pi|h})\!-\!v(x_{\pi|h-1},\mathbf{0}_{-\pi|h-1})\Big)\nonumber\\
    &=& v(x_{\pi|0},\mathbf{1}_{-\pi|0})-v(x_{\pi|n},\mathbf{1}_{-\pi|n})+v(x_{\pi|n},\mathbf{0}_{-\pi|n})-
    v(x_{\pi|0},\mathbf{0}_{-\pi|0})\nonumber\\
    &=& v(\mathbf{1})-v(x)+v(x)-v(\mathbf{0})=1-0=1, 
  \end{eqnarray*}
  so that $\sum_{i=1}^n \Psi_i(v)=1$, i.e., $\Psi$ is efficient.
  
  The definition of $\Psi$ is obviously anonymous, so that it is also symmetric. If agent $i\in N$ is a null player and 
  $\pi\in\mathcal{S}_n$ arbitrary, then $v(\dnim)=v(\dni)$ and  $v(\upim)=v(\upi)$, so that $\Psi_i(v)=0$, i.e., $\Psi$ 
  satisfies the null player property. Since $x+y=\max\{x,y\}+\min\{x,y\}$ for all $x,y\in\mathbb{R}$ and due to the linearity of 
  finite sums and integrals, $\Psi$ also satisfies the transfer axiom.
\end{proof}

Since the proof of Proposition~\ref{proposition_power_properties} actually shows that the stated properties are even satisfied for the 
summation part without integrating we conclude:
\begin{proposition}
  \label{prop_generalized_measure_point}
  For every $\alpha\in[0,1]$ the mapping $\Psi^a$, where $a=(\alpha,\dots,\alpha)\in[0,1]^n$, defined by
  $$ 
  \Psi^a_i(v) = \frac{1}{n!}\sum_{\pi\in\mathcal{S}_n}\Big(\left[v(a_{\pi_{<i}},\mathbf{1}_{\pi_{\ge i}})-v(a_{\pi_{<i}},\mathbf{0}_{\pi_{\ge i}})\right]
  -\left[v(a_{\pi_{\le i}},\mathbf{1}_{\pi_{> i}})-v(a_{\pi_{\le i}},\mathbf{0}_{\pi_{> i}})\right]\Big)
  $$
  for all $i\in N$, is positive, efficient, anonymous, symmetric, and satisfies both the null 
  player and the transfer property for interval simple games.
\end{proposition}

In other words, $\Psi^a$ does not consider all possible vote vectors with equal probability but just a specific and symmetric 
one, i.e., $a_i=a_j$ for all $i,j\in N$. This construction can easily be generalized by introducing a {\lq\lq}symmetric{\rq\rq} 
density function in the integration part of Definition~\ref{def_ssi_interval_sg}.

\begin{proposition}
  \label{prop_generalized_measure_density}
  Let $f$ be a $[0,1]^n\to\mathbb{R}_{\ge 0}$ mapping satisfying $\int_{[0,1]^n} f(x)\operatorname{d}x=1$ and 
  $f(x)=f\!\left(\left(x_{\pi(i)}\right)_{i\in N}\right)$ for all $x\in[0,1]^n$ and all $\pi\in\mathcal{S}_n$. Then, 
  the mapping $\Psi^f$ defined by
  $$
    \Psi^f_i(v) = \frac{1}{n!}\sum_{\pi\in\mathcal{S}_n}\int_{[0,1]^n} f(x)\cdot\left(\left[v(x_{\pi_{<i}},\mathbf{1}_{\pi_{\ge i}})-v(x_{\pi_{<i}},\mathbf{0}_{\pi_{\ge i}})\right]
  -\left[v(x_{\pi_{\le i}},\mathbf{1}_{\pi_{> i}})-v(x_{\pi_{\le i}},\mathbf{0}_{\pi_{> i}})\right]\right)\operatorname{d}x
  $$
  for all $i\in N$, is positive, efficient, anonymous, symmetric, and satisfies both the null 
  player and the transfer property for interval simple games.
\end{proposition}

The example from Proposition~\ref{prop_generalized_measure_point} can be interpreted in the context of the construction from 
Proposition~\ref{prop_generalized_measure_density} by using a Dirac measure. In \cite{kurz2014measuring} it was conjectured 
that every power index for interval simple games that satisfies symmetry, efficiency, the null player property, and the transfer 
property coincides with $\Psi$. However, the two constructions above show that this is wrong. To this end we consider the specific 
interval simple game defined by $v(x)=x_1x_2^2$ for $n\ge 2$ players. It can be easily checked that
$$
  \Psi(v)=\left(\frac{5}{12},\frac{7}{12},0,\dots,0\right)
  \quad\text{and}\quad
  \Psi^a(v)=\left(\frac{1}{2}-\frac{\alpha-\alpha^2}{2},\frac{1}{2}+\frac{\alpha-\alpha^2}{2},0,\dots,0\right)
$$
for $a=(\alpha,\dots,\alpha)$ with $\alpha\in[0,1]$. For $\alpha\neq \tfrac{1}{2}\pm \tfrac{1}{2\sqrt{3}}$ we have 
$\Psi(v)\neq \Psi^a(v)$. In the setting of Proposition~\ref{prop_generalized_measure_density} we can even find way more 
density functions $f$ with $\Psi(v)\neq \Psi^f(v)$ than those already mentioned. We remark that the underlying idea of the constructions 
of Proposition~\ref{prop_generalized_measure_point} and Proposition~\ref{prop_generalized_measure_density} can also be applied 
to $(j,k)$ simple games. For simple games the analogy of Proposition~\ref{prop_generalized_measure_point} is the roll call model 
where either all players say {\lq\lq}yes{\rq\rq} or all players say {\lq\lq}no{\rq\rq}, corresponding to $a=\mathbf{0}$ and 
$a=\mathbf{1}$, respectively. The analogy of Proposition~\ref{prop_generalized_measure_density} for simple games 
is the roll call model with exchangeable probabilities for vote vectors $x\in\{0,1\}^n$ as proven in \cite{hu2006asymmetric}. 
In the light of the characterization result from \cite{kurz2018roll} it would be interesting to know whether the parametric 
set of examples from Proposition~\ref{prop_generalized_measure_density} can be even further generalized. Indeed further examples 
satisfying symmetry, efficiency, the null player property, and the transfer property exist like e.g.\ $v\mapsto\Psi(v^2)$, 
where $v^2(x)=v(x)^2$ for all $x\in[0,1]^n$, which is an interval simple game provided that $v$ is an interval simple game. 
For $(j,k)$ simple games we remark that efficiency, symmetry, the null player and the transfer property are sufficiently to determine the 
Shapley-Shubik index for $(j,2)$ simple games, see \cite[Theorem 5.1]{freixas2005shapley}, while for $(j,k)$ simple 
games with $k>2$ further axioms are needed, see \cite[Theorem 1.3]{shapleybookjosep}. Additionally axioms for a characterization 
of $\Psi$ for interval simple games are given in Section~\ref{sec_axiomatization}.  

\medskip

As shown in the previous subsection for $(j,k)$ simple games, we can easily simplify the definition of $\Psi$:
\begin{proposition}
  \label{prop_formula_ssi_interval_sg}
  For every interval simple game $v$ with player set $N$ and every player $i\in N$ we have
  \begin{equation}
  \label{eq_ssi_interval_sg}
  \Psi_i(v)=\sum_{i\in S\subseteq N} \frac{(s-1)!(n-s)!}{n!}\cdot\left[C(v,S)-C(v,S\backslash\{i\})\right],
  \end{equation}
  where $C(v,T)=\int_{[0,1]^n} v(\mathbf{1}_T,x_{-T})-v(\mathbf{0}_T,x_{-T})\,\operatorname{d}x$ for all $T\subseteq N$.
\end{proposition}
\begin{proof}
  Setting $S_\pi^i:=\{j\in N\,:\,\pi(j)\ge \pi(i)\}$ we have $S_\pi^i=S$ for exactly $(s-1)!(n-s)!$ permutations $\pi\in\mathcal{S}_n$ and 
an arbitrary set $\{i\}\subseteq S\subseteq N$, so that Equation~(\ref{eq_ssi_interval_sg}) is just a simplification of 
Equation~(\ref{eq_def_ssi_interval_sg}).
\end{proof}

Now we have to catch up on the fact that $\Psi$ is well-defined, i.e., all stated integrals indeed exist. Using Fubini's theorem 
we can easily switch between the formulation in Definition~\ref{def_ssi_interval_sg} and that of Proposition~\ref{prop_formula_ssi_interval_sg}, 
so that it suffices to show the existence of
$$
  \int_0^1\dots\int_0^1 v(\mathbf{1}_T,x_{-T}) \operatorname{d}x_1\dots \operatorname{d}x_n 
  \quad\text{and}\quad
  \int_0^1\dots\int_0^1 v(\mathbf{0}_T,x_{-T}) \operatorname{d}x_1\dots \operatorname{d}x_n 
$$
for all $T\subseteq N$. Here we can use the fact that $v$ is monotone in each coordinate, so that this property remains true 
if we iteratively integrate one coordinate after the other. A monotone $[0,1]\to[0,1]$ function is obviously integrable. 

So, the monotonicity of interval simple games $v$ is a sufficient condition for the existence of the integrals. Another reason 
is that we eventually lose the property that $\Psi$ is positive if we allow functions $v$ that are not (weakly) monotone 
increasing. The very same effect also happens for simple games. As an example let $n=3$ and $v$ map the coalitions $\{1\}$, 
$\{1,3\}$, $\{1,2,3\}$ to $1$ and all other coalitions to zero. Using Equation~(\ref{eq_ssi_simple_games}) we obtain 
$\operatorname{SSI}_2(v)=-\tfrac{1}{3}<0$. 


\section{An axiomatization for the Shapley-Shubik like index for interval decisions}
\label{sec_axiomatization}

As we have seen in Proposition~\ref{prop_generalized_measure_point} and Proposition~\ref{prop_generalized_measure_density}, for interval simple 
games the axioms of efficiency (E), symmetry (S), the null player (NP), and the transfer property (T) are not sufficient to uniquely characterize a 
power index. Also anonymity (A) and positivity (P) are satisfied by our parametric examples of power indices. So, for an axiomatization we need 
some further axioms. To that end we consider a special class of step functions. In general, a \emph{step function} 
is a function that takes only finitely many values. Here we consider only regular pavings based on rectangular boxes. To formalize 
this, we set 
$$
  \mathcal{D}_{p}=\left\{ \left( \alpha _{0},\alpha _{1},...,\alpha_{p}\right) \in \left[ 0,1\right] ^{p+1}
  \,:\, \alpha _{0}=0,\alpha _{p}=1\text{ and    }\alpha _{i}<\alpha _{i+1}\text{ for }i=0,1,...,p-1\right\}
$$
for a given integer $p\ge 1$ and $\mathcal{A}_{p}=\left\{0,\tfrac{1}{2},1,\tfrac{3}{2},\dots,p-\tfrac{1}{2},p\right\}^{n}$, where $n$ denotes the number of players.   
Given $\alpha \in \mathcal{D}_{p}$, we denote by $\omega \left(\alpha\right) =
\max_{1\leq h\leq p}\left( \alpha _{h}-\alpha _{h-1}\right)$ the
maximal difference between two consecutive $\alpha_{h}$. For each $e 
\in \mathcal{A}_{p}$ and each $\alpha \in \mathcal{D}_{p}$ by $(\alpha)_e$ we abbreviate the open box $I_1\times I_2\times\dots\times I_n$, 
where the intervals are given by $I_j=\{\alpha_{e_j}\}$ if $e_j\in\mathbb{N}$ and $I_j=\left(\alpha_{e_j-1/2},\alpha_{e_j+1/2}\right)$ 
otherwise.\footnote{Slightly abusing notation, we may also write $I_j=\left(\alpha_{\left\lfloor e_j\right\rfloor},\alpha_{\left\lceil e_j\right\rceil}\right)$.} 
We denote the closure of $(\alpha)_e$ by $[\alpha]_e$, i.e., all open intervals $I_j$ are replaced by the corresponding closed intervals.
  
Given a paving of $[0,1]^n$, described by $\alpha\in\mathcal{D}_p$, we assume that a step function is constant on the interior of every face of one of the boxes.  
\begin{definition}
  \label{definition_step_function}
  A function $f\colon[0,1]^n\to[0,1]$ is a \emph{step function} if there exists an integer $p\ge 1$ and $\alpha\in \mathcal{D}_p$ such 
  that $f$ is constant on $(\alpha)_e$ for all $e\in\mathcal{A}_{p}$. We 
  call $\alpha$ a \emph{discretization} of $f$.   
\end{definition}

Figure~\ref{phase_2}, in the appendix, shows an example of a step function for $n=2$ players and discretization $\alpha=\left(0,\tfrac{1}{4},1\right)$, i.e., $p=2$. 
The open rectangle $\left(0,\tfrac{1}{4}\right)\times\left(\tfrac{1}{4},1\right)$ is denoted by $(\alpha)_e$, where $e=\left(\tfrac{1}{2},\tfrac{3}{2}\right)\in\mathcal{A}_2$. 
The (closed) edge $\left[0,\tfrac{1}{4}\right]\times\{1\}$ is denoted by $[\alpha]_e$, where $e=\left(\tfrac{1}{2},2\right)$, and the vertex 
$\left\{1\right\}\times\left\{\tfrac{1}{4}\right\}$, i.e., the point with coordinates $\left(1,\tfrac{1}{4}\right)$, is denoted by either $[\alpha]_e$ or 
$(\alpha)_e$, where $e=\left(2,1\right)$. We remark that the number of fractional entries in $e\in\mathcal{A}_p$ gives the geometric dimension of the 
face $[\alpha]_e$ (or $(\alpha)_e$). In order to have a simple notation for the full dimensional boxes we set 
$\overline{\mathcal{A}}_p=\left\{\tfrac{1}{2},\tfrac{3}{2},\dots,p-\tfrac{1}{2}\right\}^{n}$. 

Since it would be beneficial to describe a step function by specifying its values on the full dimensional boxes, i.e., on $(\alpha)_{e}$ for 
all $e\in\overline{\mathcal{A}}_p$, we introduce the concept of a \emph{regular step function}. Here we determine the values of the step 
function on the faces that are not full dimensional simply by averaging over the values attained in the interior of all 
neighbored full dimensional boxes. This approach allows us to embed any given $(j,k)$ simple game in an interval simple game, see 
Definition~\ref{def_natural_embedding}, so that the corresponding power indices coincide, see Proposition~\ref{prop_embed}. 

\begin{definition}
  \label{definition_regular_step_function}
  A step function $f\colon[0,1]^n\to[0,1]$ with discretization $\alpha\in \mathcal{D}_p$ is called \emph{regular} if we have $f(\mathbf{0})=0$, $f(\mathbf{1})=1$, and
  \begin{equation}
    \label{eq_average}
    f(x)=\sum_{e\in E(x)} f(c_e) / |E(x)|\quad \forall x\in[0,1]^n\backslash\{\mathbf{0},\mathbf{1}\}, 
  \end{equation} 
  where $c_e$ denotes the center\footnote{\label{fn_center}The center of $[\alpha]_e$ is given by 
  $\left(\frac{\alpha_{\left\lfloor e_1\right\rfloor}+\alpha_{\left\lceil e_1\right\rceil}}{2},\dots,
  \frac{\alpha_{\left\lfloor e_n\right\rfloor}+\alpha_{\left\lceil e_n\right\rceil}}{2}\right)$.} of $[\alpha]_e$ 
  and $E(x)$ denotes the set of elements $e$ in $\overline{\mathcal{A}}_p$ such that $x\in[\alpha]_e$. If Equation~(\ref{eq_average}) 
  is only satisfied for all $x\in(0,1)^n$, then $f$ is called \emph{semi-regular}.
\end{definition}  

We need to specify a regular step function $f$ only on the interior of the full dimensional boxes $(\alpha)_e$, where $e\in\overline{\mathcal{A}_p}$. Since 
a step function is constant on those faces it is e.g.\ sufficient to specify the value $f(c_e)$ at the corresponding center. For all other points in $[0,1]^n$, 
i.e., those that are on the boundary of at least one of the boxes, we can use Equation~(\ref{eq_average}) to determine the function value. This average type extension to 
the boundary, is essential in our context, since $\Psi(v)$ heavily depends on evaluations of $v$ at the boundary 
of $[0,1]^n$.\footnote{We will see latter that boundary points of boxes of the underlying paving that are not boundary points of $[0,1]^n$ do not play 
a role for the value of $\Psi$.} Alternatively, we may replace $0$ and $1$ by $\varepsilon$ and $1-\varepsilon$, respectively, and consider the limit $\varepsilon\to 0$. 
Here we prefer the less technical, but more restrictive, variant of Equation~(\ref{eq_average}).

Note that if a regular step function $f$ is monotone, with given parameters $p$ and $\alpha$, then $f(x)\le f(y)$ is equivalent to $e_x\le e_y$ 
for all $e_x,e_y\in\overline{\mathcal{A}}_p$ and all $x\in(\alpha)_{e_x}$, $y\in(\alpha)_{e_y}$.  

Next we show that we can embed each $(j,k)$ simple game $v$ as interval simple game $\hat{v}$ and that the corresponding power indices 
$\operatorname{SSI}(v)$ and $\Psi(\hat{v})$ coincide.

\begin{definition}
  \label{def_natural_embedding}
  For integers $j,k\ge 2$ let $v$ be a $(j,k)$ simple game for $n$ players. For $\alpha=\left(0,\tfrac{1}{j},\tfrac{2}{j},\dots,\tfrac{j-1}{j},1\right)$ 
  let $\hat{v}$ be the regular step function with discretization $\alpha$ uniquely defined by $\hat{v}(x)=v(e)/(k-1)$ for all $x\in (\alpha)_{\bar{e}}$ 
  and all $\bar{e}\in\overline{\mathcal{A}}_j$, where $e=\bar{e}-\tfrac{1}{2}\cdot\mathbf{1}$. We call $\hat{v}$ the natural embedding of $v$. 
\end{definition}

It can be easily checked that this embedding transfers the null player property and the property of symmetric players, i.e., a null player in $v$ is also a null player 
in $\hat{v}$ and two symmetric players in $v$ are also symmetric in $\hat{v}$. The $[0,1]^n\to[0,1]$ function $\hat{v}$ is monotone and satisfies $\hat{v}(\mathbf{0})=0$ 
and $\hat{v}(\mathbf{1})=1$,\footnote{Since $v(0,\dots,0)=0$ and $v(j-1,\dots,j-1)=k-1$, the conditions $f(\mathbf{0})=0$, $f(\mathbf{1})=1$ in 
Definition~\ref{definition_regular_step_function} are not necessary for this conclusion.} i.e., $\hat{v}$ is an interval simple game. 

\begin{proposition}
  \label{prop_embed}
  For integers $j,k\ge 2$ let $v$ be a $(j,k)$ simple game for $n$ players and $\hat{v}$ its natural embedding with discretization $\alpha$, as 
  specified in Definition~\ref{def_natural_embedding}, then $\operatorname{SSI}(v)=\Psi(\hat{v})$.
\end{proposition}
\begin{proof} 
  Due to (\ref{eq_roll_call_jk_simple_games_simplified}) and Equation~(\ref{eq_ssi_interval_sg}) it suffices to verify the coincidence of the 
  two different expressions for $C(v,T)$ and $C(\hat{v},T)$ for all $T\subseteq N$. We compute
  \begin{eqnarray*}
    C(\hat{v},T)&=&\int_{[0,1]^n} \hat{v}(\mathbf{1}_T,x_{-T})-\hat{v}(\mathbf{0}_T,x_{-T})\,\operatorname{d}x\\
    &=& \sum_{\overline{e}\in\overline{\mathcal{A}}_j} \int_{(\alpha)_{\overline{e}}} \hat{v}(\mathbf{1}_T,x_{-T})-\hat{v}(\mathbf{0}_T,x_{-T})\,\operatorname{d}x\\ 
    &=& \frac{1}{k-1}\cdot \sum_{\overline{e}\in\overline{\mathcal{A}}_j} \int_{(\alpha)_{\overline{e}}} v(\mathbf{1}_T,e_{-T})-v(\mathbf{0}_T,e_{-T})\,\operatorname{d}x 
    \text{ where } e=\bar{e}-\tfrac{1}{2}\cdot\mathbf{1}\\
    &=& \frac{1}{j^n(k-1)}\cdot \sum_{x\in J^n} v(\mathbf{1}_T,x_{-T})-v(\mathbf{0}_T,x_{-T})  
    = C(v,T).
  \end{eqnarray*}
\end{proof}

So, in other words we have associated a $(j,k)$ simple game with an interval simple game that is a step function. For $j=2$ there is an even more 
general statement. For any parameter $\tau\in(0,1)$ we can replace the discretization $\alpha=\left(0,\tfrac{1}{2},1\right)$ by $(0,\tau,1)$. 
\begin{proposition}
  \label{prop_embed_j_2}
  For an integer $k\ge 2$ and $\tau\in(0,1)$ let $v$ be a $(2,k)$ simple game for $n$ players and $\hat{v}$ be the regular step function with discretization 
  $\alpha=(0,\tau,1)$ uniquely defined by $\hat{v}(x)=v(e)/(k-1)$ for all $x\in (\alpha)_{\bar{e}}$ and all $\bar{e}\in\overline{\mathcal{A}}_2$, 
  where $e=\bar{e}-\tfrac{1}{2}\cdot\mathbf{1}$. Then, $\operatorname{SSI}(v)=\Psi(\hat{v})$. 
\end{proposition}
\begin{proof} 
  Let $\tilde{v}$ be a TU game, i.e., a mapping $2^N\to\mathbb{R}$, which maps $\emptyset$ to $0$. We associate $\tilde{v}$ with $v$ by setting $\tilde{v}(S)=v(\tilde{S})/(k-1)$ 
  for all $S\subseteq N$ where $\tilde{S}\in \{0,1\}^n$ is defined for all $i\in N$ by $\tilde{S}_i=1$ if $i\in S$ and $\tilde{S}_i=0$ otherwise. With this, $\hat{v}$ is the regular step function with discretization $\alpha=(0,\tau,1)$ uniquely defined by 
  $\hat{v}(x)=\tilde{v}(e)$ for all $x\in (\alpha)_{\bar{e}}$ and all $\bar{e}\in\overline{\mathcal{A}}_2$, where $e=\bar{e}-\tfrac{1}{2}\cdot\mathbf{1}$. If 
  $v$ is such that $\tilde{v}(S)=1$ if $T\subseteq S$ and $\tilde{v}(S)=0$ otherwise, where $\emptyset\neq T\subseteq N$, i.e., 
  $\tilde{v}$ is a unanimity game, then the null player property, symmetry, and efficiency give $\operatorname{SSI}_i(v)=\Psi_i(\hat{v})=1/|T|$ for all $i\in T$ and $\operatorname{SSI}_i(v)=\Psi_i(\hat{v})=0$ otherwise. Note 
  that the values are independent from $\tau$. Since each TU game can be written as a linear combination of unanimity games and $\Psi$ is linear, also in general, 
  $\Psi(\hat{v})$ does not depend on $\tau$. Using $\tau=\tfrac{1}{2}$ we can use Proposition~\ref{prop_embed} to conclude $\operatorname{SSI}(v)=\Psi(\hat{v})$.
\end{proof}

We remark that the result is essentially implied by the roll call interpretation for the Shapley value for TU games, or for $k=2$ for the Shapley-Shubik 
index for simple games, with a probability of $\tau$ for voting {\lq\lq}no{\rq\rq} and a probability of $1-\tau$ for voting {\lq\lq}yes{\rq\rq}, see e.g.\ 
\cite{kurz2018roll}. The finer coincidence of $C(\hat{v},T)=C(v,T)$ is valid for $\tau=\tfrac{1}{2}$ only. 

We remark that in \cite[Definition 3.7]{freixas2005shapley} the author considers a more general definition of a Shapley-Shubik like index for 
$(j,k)$ simple games than we have presented here. In general a so-called numeric evaluation comes into play and our case is called uniform 
numeric evaluation. We remark that Proposition~\ref{prop_embed} can also be generalized in that direction modifying the paving for the 
associated step function in a natural way.  

The proof of Proposition~\ref{prop_embed_j_2} suggests an even more general statement for $(2,k)$ simple games. The essential part is that we used 
a map $\eta$ from the set of $(2,k)$ simple games into interval simple games that preserves null players, symmetric players and linearity, i.e., 
$\eta\left(\sum_T \lambda_T\cdot v_T\right)=\sum_T \lambda_T \cdot\eta\left(v_T\right)$. For each such map $\eta$ we have $\operatorname{SSI}(v)=\Psi(\eta(v))$ 
since we can write $v$ as a linear combination of unanimity unanimity games. As an application we mention the following embedding of a simple game as an 
interval simple game, which we will use later on, see the proof of Lemma~\ref{lemma_hsi_constants}.

\begin{proposition}
  \label{prop_embed_sg_special}
  Let $v$ be a simple game for $n$ players and $\tilde{v}$ be the semi-regular step function (with discretization $\alpha=(0,1)$) 
  defined by $\tilde{v}(x)=1$ if $\{i\,:\,x_i=1\}$ is a winning coalition in $v$ and $\tilde{v}(x)=0$ otherwise. 
  Then, $\operatorname{SSI}(v)=\Psi(\tilde{v})$. 
\end{proposition}

\medskip

Our aim is to approximate interval simple games by regular step functions. Given an interval simple game $v\colon[0,1]^n\to[0,1]$ and a 
discretization $\alpha\in\mathcal{D}_p$ we call a regular step function $f\colon[0,1]^n\to[0,1]$ with discretization $\alpha$ an 
\emph{approximation} of $v$ if for each $e\in\overline{\mathcal{A}}_p$ there exists a value $x_e\in[\alpha]_e$ such that $f(y_e)=v(x_e)$ for all 
$y_e\in(\alpha)_e$. Of course it is easy to construct such approximations, i.e., we may take the center $x_e=c_e$ of $[\alpha]_e$, 
see Footnote~\ref{fn_center} . Any approximation $f$ of $v$ is an interval simple game, i.e., it is monotone and we 
have $f(\mathbf{0})=0$, $f(\mathbf{1})=1$. Moreover, if players $i$ and $j$ are symmetric in $v$, then they are symmetric in $f$ and if player 
$i$ is a null player in $v$, then it is a null player in $f$. The condition $f(y_e)=v(x_e)$ captures the idea of an approximation.\footnote{Alternatively, 
we might have also included the supremum or the infimum on $[\alpha]_e$, or some value in between, as a possible value for $f(y_e)$ or require that 
the approximation get {\lq\lq}better{\rq\rq} if $\omega(\alpha)$ 
tends to zero.}

\begin{definition}
  \label{def_approx}
  Let $v$ be an interval simple game and $\left(f^h\right)_{h\in\mathbb{N}}$ a sequence of approximations with discretizations $\alpha^h$. 
  If $\lim_{h\to\infty}\omega\!\left(\alpha^h\right)=0$ and $\lim_{h\to\infty} \sup_{x\in[0,1]^n} \left|v(x)-f^h(x)\right|=0$, then $\left(f^h\right)_{h\in\mathbb{N}}$ 
  is called an \emph{approximation sequence} of $v$. If a least one approximation sequence exists then $v$ is called \emph{approximable}.  
\end{definition}

Note that we have required point-wise convergence in Definition~\ref{def_approx}. In combination with Equation~(\ref{eq_average}) for regular step functions, 
this is quite restrictive but e.g.\ satisfied by interval simple games that are continuous. Technically, we might relax the condition of point-wise convergence 
so that deviations that do not change the value of $\Psi(v)$ are ignored. To ease the exposition we do not go into details here. 

In the subsequent Subsection~\ref{subsec_HIS} we introduce a further axiom in Definition~\ref{def_his} that allows to characterize $\Psi$ on regular step functions, see 
Theorem~\ref{thm_characterization_step_function}. In Subsection~\ref{subsec_characterization} we introduce another axiom in Definition~\ref{definition_limit_psi} 
that allows us to characterize $\Psi$ on approximable interval simple games, see Theorem~\ref{thm_characterization}. 

\subsection{Homogeneous Increments Sharing}
\label{subsec_HIS}

Let $v$ be an interval simple game and $S\in 2^{N}\backslash \{N\}$ a coalition. Then, the \emph{potential influence}  
of the coalition $S$ denoted by $\Delta v\left(S,x_{-S}\right)$ is defined by $\Delta v\left(
S,x_{-S}\right) =v\left( \mathbf{1}_{S}\,,x_{-S}\right) -v\left( \mathbf{0}_{S}\,,x_{-S}\right) $. The potential 
influence $\Delta v\left(S,x_{-S}\right)$ of $S$ measures the greatest change in the social decision that may be 
observed when voters in $S$ change their respective opinions assuming that the profile of the voters in $N\backslash S$ 
is given by $x_{-S}$. A $T$-domain is a Cartesian product $D=\bigtimes_{i\in T}\left[a_{i},b_{i}\right] $ given some 
$a_{i},b_{i}\in [0,1]$, meaning that each voter $i\in T$ freely and independently chooses his levels of approbation 
from $\left[a_{i},b_{i}\right]\subseteq [0,1]$.

\begin{definition}
  \label{definition_local_improvement}
  Let $v$ be an interval simple game for $n$ players, $\emptyset\subsetneq S\subsetneq N$ 
  a coalition, $\varepsilon\in\mathbb{R}_{\ge 0}$,   
  and $D=\bigtimes_{i\in N\backslash S}\left[a_{i},b_{i}\right]$ be an $(N\backslash S)$-domain, where $0\le a_i \le b_i\le 1$ 
  for all $i\in N\backslash S$. If an interval simple game $u$ for $n$ players satisfies
  \begin{itemize}
  \item $\forall x_{-S}\in (0,1)^{N\backslash S},$ $\Delta v\left(S,x_{-S}\right) =\left\{
\begin{tabular}{lll}
$\Delta u\left( S\,,x_{-S}\right) $ & \text{if} & $x_{-S}\notin
\bigtimes_{i\in N\backslash S}\left[a_{i},b_{i}\right] $ \\
&  &  \\
$\Delta u\left( S,x_{-S}\right) +\varepsilon $ & \text{if} & $
x_{-S}\in \bigtimes_{i\in N\backslash S}\left( a_{i},b_{i}\right) $
\end{tabular}
\right.$
\item $\forall T\in 2^{N}\backslash \left\{S\right\} 
,\forall x_{-T}\in (0,1)^{N\backslash T}$, $\Delta v\left( T,x_{-T}\right) =\Delta
u\left( T,x_{-T}\right)$,
\end{itemize}
then $v$ is a \emph{local increment} of $u$ and we write $u\localinc{S}{\varepsilon}{D}v$. For 
$S=\emptyset$ we write $u\localinc{\emptyset}{\varepsilon}{D}v$ if $v(x)=u(x)+\varepsilon$ for 
all $x\in(0,1)^n$, i.e., we are indirectly setting $D=[0,1]^n$. For $S=N$ we should have chosen $D=\emptyset$, but 
we choose $D=[c,1]^n$ instead, where $c\in(0,1)$. With this, the condition for $u\localinc{S}{\varepsilon}{D}v$ 
is $v(x)=u(x)+\varepsilon$ for all $x\in(c,1)^n$.
\end{definition}

In words $u\localinc{S}{\varepsilon}{D}v$ means that on the one hand, the potential influence of coalition $S$ increases by a
constant increment $\varepsilon$ whenever each voter $i\in N\backslash S$ picks his
opinion from $\left(a_{i},b_{i}\right)$, but remains unchanged if the opinion of at least one voter $i\in N\backslash S$ 
is outside of  $\left[a_{i},b_{i}\right]$. It is then reasonable that the corresponding increment in the collective decision
mainly comes from voters in $S$ and is uniform, local and elsewhere valid on $\bigtimes_{i\in S}\left(a_{i},b_{i}\right)$. 
On the other hand, the potential influence of any other coalition $T$ remains unchanged unless some
voters in $S$ show a full support ($x_{i}=1$), or no support ($x_{i}=0$). In
such situations, the shares by a conceivable power index from $u$ to $v$ are expected to change accordingly by only uniformly rewarding 
voters in $S$ in the expense of voters outside of $S$. In the extreme cases of $S=\emptyset$ and $S=N$, this reasoning does not 
makes sense and we consider uniform changes for a completely symmetric $D$. 

\begin{definition}
  \label{def_his}
  A power index $\Phi$ for interval simple games satisfies the \emph{homogeneous increments sharing} (HIS) axiom, if for all 
  $\emptyset\subseteq S\subseteq N$ and for all interval simple games $u$ and $v$ such that 
  $u\localinc{S}{\varepsilon}{D}v$ for some $\varepsilon >0$ and some $(N\backslash S)$-domain $D$ we have
\begin{equation}
\Phi_{i}\left( v\right) -\Phi_{i}\left( u\right) =\left\{
\begin{array}{ccc}
\lambda_{\Phi }(S)\cdot \varepsilon\cdot \operatorname{vol}(D) & \text{if} & i\in S, \\
&  &  \\
-\gamma_{\Phi }(S)\cdot \varepsilon \cdot \operatorname{vol}(D) & \text{if} & i\notin S,
\end{array}%
\right.  \label{Eq HIS}
\end{equation}%
where $\lambda_{\Phi}(S)$ and $\gamma_{\Phi}(S)$ are two real 
constants that do only depend on $S$, i.e., they do neither depend on $u$ and $v$ nor on $\varepsilon$ and $D$, and 
$\operatorname{vol}(D)$ denotes the volume of $D$. (Note the special shapes we assume for $D$ in the 
case of $S=\emptyset$ or $S=N$.)
\end{definition}


The term $\varepsilon\cdot\operatorname{vol}(D)$ captures the fact that the change in the share of a voter is both proportional to the
magnitude $\varepsilon$ of the homogeneous increment and to the (local) volume $\operatorname{vol}(D)$ of the domain on which this change 
occurs. 
For simple games the analog of (HIS) is the axiom of \emph{Symmetric Gain-Loss} (SymGL) \cite[p. 93]{laruelle2001shapley}. 
The second part cannot occur for simple games since we have $v(\emptyset)=0$ and $v(N)=1$ for every simple game by definition. For 
$(j,k)$ simple games the axiom of \emph{level change on unanimity games}, see \cite[Subsection 1.8.2]{shapleybookjosep}, is closely related to (HIS).

\begin{proposition}
  \label{proposition_HIS} 
  The power index $\Psi$ for interval simple games satisfies (HIS) for
  \begin{equation}
  \left( \lambda_{\Psi}(S),\gamma_{\Psi}(S)\right) =
  \left( \frac{(s-1)!(n-s)!}{n!},\frac{s!(n-s-1)!}{n!}\right)  \label{eq_coeff_HIS}
  \end{equation}
  for all $\emptyset\subsetneq S\subsetneq N$. Moreover, we can set $\left( \lambda_{\Psi}(\emptyset),\gamma_{\Psi}(\emptyset)\right)=
  \left( \lambda_{\Psi}(N),\gamma_{\Psi}(N)\right)=(0,0)$. 
\end{proposition}
\begin{proof}
Let $u$, $v$, $S$, $\varepsilon$, and $D$ be given such that $u\localinc{S}{\varepsilon}{D}v$ and $S\notin\{\emptyset,N\}$. 
Due to the formula for $\Delta v(\cdot,\cdot)$ in Definition~\ref{definition_local_improvement}, we have
$$
  C(v,S)=\int_{[0,1]^n} v(\mathbf{1}_S,x_{-S})-v(\mathbf{0}_S,x_{-S})=C(u,S)+\varepsilon\cdot\operatorname{vol}(D)
$$
and
$$
  C(v,T)=\int_{[0,1]^n} v(\mathbf{1}_T,x_{-T})-v(\mathbf{0}_T,x_{-T})=C(u,T)
$$
for all $T\in 2^N\backslash\{S\}$. (For $T=\emptyset$ we have $C(v,\emptyset)=0=C(u,\emptyset)$ and for $T=N$ we have 
$C(v,N)=1=C(u,N)$.) 
From Equation~(\ref{eq_ssi_interval_sg}) we then conclude $\lambda_{\Psi}(S)=\frac{(s-1)!(n-s)!}{n!}$ for the cases $i\in S$ 
and $\gamma_{\Psi}(S)=\frac{s!(n-s-1)!}{n!}$ for the cases $j\notin S$, where only $C(u,S\cup\{j\})$ is different from $C(v,S\cup\{j\})$. 

For the second part, let $u$, $v$, $S$, $\varepsilon$, and $D$ be given such that $u\localinc{S}{\varepsilon}{D}v$ and $S\in\{\emptyset,N\}$. Since 
$\Psi$ is linear, anonymous, and efficient we have $\Psi(u)=\Psi(v)$.
\end{proof}

We remark that $\Psi$ also satisfies (HIS) for negative parameters $\varepsilon$, which is the same as interchanging the roles of $u$ 
and $v$ and considering $-\varepsilon$ instead. 

\medskip

In the context of simple games a local increment from a simple game $u$ to another simple game $v$ means that the set 
of winning coalitions of $v$, i.e., $S\subseteq N$ with $v(S)=1$, consists of the set of winning coalitions of $u$ and 
an additional winning coalition that was losing in $u$. The effects on the number of swing coalitions, i.e., 
those with $v(S)-v(S\backslash \{i\})=1$, when removing one minimal winning coalition from a simple game are well known, 
see e.g.~\cite[Lemma 3.3.12]{felsenthal1998measurement}. To be more precise, let $u$ and $v$ be two simple games such that 
the winning coalitions of $v$ are given by the winning coalitions of $u$ and a coalition $\emptyset
\subsetneq S\subsetneq N$ that is losing in $u$. As notation we write $%
v=u\oplus S$. For all $i\in S$ and all $j\in N\backslash S$ we have
\begin{equation}
\operatorname{SSI}_{i}(v)=\operatorname{SSI}_{i}(u)+\frac{(s-1)!(n-s)!}{n!}\quad \text{and}%
\quad
\operatorname{SSI}_{j}(v)=\operatorname{SSI}_{j}(u)-\frac{s!(n-s-1)!}{n!},
\label{eq_SSI_add_S}
\end{equation}%
respectively. 
Note that we cannot choose $S=\emptyset$ or $S=N$ in that setting.

\begin{lemma}
\label{lemma_hsi_constants} If $\Phi $ is a power index for interval simple
games that are semi-regular step functions that simultaneously satisfies
(E), (NP), and (HIS) and $n\geq 3$, then we have
\begin{equation}
\left( \lambda _{\Phi }(S),\gamma _{\Phi }(S)\right) =\left( \frac{%
(s-1)!(n-s)!}{n!},\frac{s!(n-s-1)!}{n!}\right)  \label{eq_hsi_constants}
\end{equation}%
for all $\emptyset \subsetneq S\subsetneq N$. Moreover, we have $\lambda
_{\Phi }(N)=\gamma _{\Phi }(\emptyset )=0$.
\end{lemma}

\begin{proof}
To each simple game $v$ we associate a semi-regular step function $\widetilde{v}$ 
via $\widetilde{v}\left( x\right) =1$ if $\left\{i\in N:x_{i}=1\right\} $ is a
winning coalition in $v$ and $\widetilde{v}\left( x\right) =0$ otherwise; c.f.\ 
Proposition~\ref{prop_embed_sg_special}.  

In the remaining part of the proof we will consider steps $v=u\oplus S$ for 
simple games $u$ and $v$, where we denote the corresponding interval simple
games by $\widetilde{u}$ and $\widetilde{v}$, respectively. We will prove
Equation~(\ref{eq_hsi_constants}) by induction from $s=n-1$ to $s=1$.
Moreover we show at each induction stage $s$ that $\operatorname{SSI}(v)=\Phi (%
\widetilde{v})$ whenever all winning coalitions in $v$ are of cardinality
greater or equal to $s$.

First note that for $v=u\oplus S$, going from $\widetilde{u}$ to $\widetilde{%
v}$ we apply a local increment construction by choosing coalition $S$ and
defining $D={\huge \times }_{i\in N\backslash S}\left[ a_{i},b_{i}\right] $ via $%
a_{i}=0$, $b_{i}=1$ for all $i\in N\backslash S$, so that
$\operatorname{vol}(D)=1$.
Choosing $\varepsilon =1$, we can easily check that $\widetilde{u}\localinc{
S}{\varepsilon}{D}\widetilde{v}$. (We have for all $%
x_{N\backslash S}\in {\huge \times }_{i\in N\backslash S}(0,1) $, $\Delta
\widetilde{u}(S,x_{-S})=\widetilde{u}(\mathbf{1}_{S},x_{-S})-\widetilde{u}(%
\mathbf{0}_{S},x_{-S})=0$ since $S$ is losing in $u$; $\Delta \widetilde{v}%
(S,x_{-S})=\widetilde{v}(\mathbf{1}_{S},x_{-S})-\widetilde{v}(\mathbf{0}%
_{S},x_{-S})=1-0$ since $S$ is winning in $v$. For any other coalition $%
T\neq S$, $T$ is winning in $v$ if and only if $T$ is winning in $u$. Thus, 
for all $x_{N\backslash T}\in {\huge \times }_{i\in N\backslash T}(0,1) $, $%
\Delta \widetilde{v}(T,x_{-T})=\Delta \widetilde{u}(T,x_{-T})$).\medskip

For a moment assume that for $u=[n;1,\dots ,1]$ and the corresponding
interval simple game $\widetilde{u}$ we have $\Phi _{i}(\widetilde{u})=%
\tfrac{1}{n}$ for all $i\in N$. Now let $v=u\oplus S$ for some coalition $%
S\subseteq N$ of cardinality $s=n-1$. Since the unique player $j$ in $%
N\backslash S$ is a null player in $\widetilde{v}$, we have $\Phi _{j}(%
\widetilde{v})=0$, so that $\gamma _{\Phi }(S)=\frac{s!(n-s-1)!}{n!}$ using
(HIS). From efficiency we then conclude $\lambda _{\Phi }(S)=\frac{%
(s-1)!(n-s)!}{n!}$. Note that $\operatorname{SSI}(u)=\Phi
(\widetilde{u})$ and by
(HIS), $\operatorname{SSI}(v)=\Phi (\widetilde{v})$. Moreover, $\operatorname{SSI}(v)=\Phi (%
\widetilde{v})$ whenever $v=u\oplus S_{1}\oplus S_{2}\oplus ...\oplus S_{p}$
for some coalitions $S_{j}$ each of cardinality $n-1$ by applying (HIS) $p$
times together with (\ref{eq_SSI_add_S}). So, the induction
start is made. \medskip

Now let $S\subseteq N$ with $0<s<n$ be given. To determine $\gamma _{\Phi
}(S)$, let $u$ be the simple game whose winning coalitions are exactly the
proper super sets of $S$. For the corresponding interval simple game $%
\widetilde{u}$ we have $\operatorname{SSI}(u)=\Phi (\widetilde{u})$
by the induction hypothesis. For $v=u\oplus S$ we have that all
players in $j\in N\backslash
S $ are null players in $v$, so that $\Phi _{j}(\widetilde{v})=\operatorname{SSI}%
_{j}(v)=0$. With this, we easily compute $\gamma _{\Phi }(S)=\Phi _{j}(%
\widetilde{u})=\operatorname{SSI}_{j}(u)=\frac{s!(n-s-1)!}{n!}$ from (HIS) and 
(\ref{eq_SSI_add_S}), which gives $\lambda _{\Phi }(S)=\frac{%
(s-1)!(n-s)!}{n!}$ using efficiency. Both $\gamma _{\Phi }(S)$ and $\lambda
_{\Phi }(S)$ depend only on $s$. Moreover, suppose that $u$ is any other
simple game whose winning coalitions are of cardinality greater than $s$ and
that $v=u\oplus S_{1}\oplus S_{2}\oplus ...\oplus S_{q}$. Then $\operatorname{SSI}%
(u)=\Phi (\widetilde{u})$ by induction hypothesis, and applying (HIS) $q$
times together with (\ref{eq_SSI_add_S}) yields $\operatorname{SSI}%
(v)=\Phi (\widetilde{v})$. \medskip

Now let us prove our assumption $\Phi _{i}(\widetilde{u})=\tfrac{1}{n}$ for
all $i\in N$, where $u=[n;1,\dots ,1]$. To this end, let $i$ and $j$ be two
arbitrary but different players in $N$ and set $X=N\backslash \{i,j\}$.
Define $\left( \varphi _{1},\dots ,\varphi _{n}\right) =\Phi (\widetilde{u})$%
. As above, player $j$ is a null player in $u\oplus \left( X\cup
\{i\}\right) $, so that (HIS) and (NP) give $\gamma _{\Phi }(X\cup
\{i\})=\varphi _{j}$. From (E) we then conclude $\lambda _{\Phi }(X\cup
\{i\})=\varphi _{j}/(n-1)$. Similarly we conclude $\gamma _{\Phi }(X\cup
\{j\})=\varphi _{i}$ and $\lambda _{\Phi }(X\cup \{i\})=\varphi _{j}/(n-1)$.
Now let $u^{\prime }=[n;1,\dots ,1]\oplus \left( X\cup \{i\}\right) \oplus
\left( X\cup \{j\}\right) $. From the above constants and (HIS) we conclude $%
\Phi _{i}(\widetilde{u^{\prime }})=\varphi _{j}/(n-1)$ and $\Phi _{j}(%
\widetilde{u^{\prime }})=\varphi _{i}/(n-1)$. In $v=u^{\prime }\oplus X%
\footnote{%
Note that this not the case for exactly two players, i.e,, $n=2$.}$ the
players $i$ and $j$ are null players so that (HIS) gives $\varphi
_{i}=\varphi _{j}$. Since $i$ and $j$ were arbitrary, we have $\Phi _{i}(%
\widetilde{u})=\tfrac{1}{n}$ for all $i\in N$ using efficiency.\medskip

For $S=\emptyset $ we can choose some $\varepsilon $ and $D$ with $%
\varepsilon >0$ and $\operatorname{vol}(D)>0$. Every player $i\in N$
is contained in $N\backslash S$, so that the value of $\lambda
_{\Phi }(\emptyset )$ does not change anything. From (E) we then
conclude $\gamma _{\Phi }(\emptyset )=0 $. For $S=N$ it is just the
other way round, i.e., every player $i\in N$ is contained in $S$ and
the value of $\gamma _{\Phi }(N)$ does not change anything.
Efficiency then gives $\lambda _{\Phi }(N)=0$.\bigskip
\end{proof}

W.l.o.g.\ we can always assume $\lambda _{\Phi }(\emptyset )=0$ and $\gamma
_{\Phi }(N)=0$. For $n=1$ the statement of Lemma~\ref{lemma_hsi_constants}
is also true, since $\emptyset$ and $N$ are the only possible subsets of $N$. 
For $n=2$ players the axioms (E), (NP), and (HIS) do not determine $\Phi$ for 
interval simple games that are semi-regular step functions as shown by the 
following parametric family.

\begin{lemma}
\label{lemma_hsi_constants 2agents} 
For $a_1,a_2\in\mathbb{R}_{\ge 0}$ with $a_1+a_2=1$ let
\begin{equation}
\Phi _{i}^{a}\left( v\right):=a_{i}+a_{j}\int_{0}^{1}\left[ v\left(
\mathbf{1}_i,\mathbf{t}_j\right) -v\left( \mathbf{0}_i,\mathbf{t}_j\right) \right] dt-a_{i}\int_{0}^{1}\left[ v\left(
\mathbf{1}_j,\mathbf{t}_i\right) -v\left( \mathbf{0}_j,\mathbf{t}_i\right) \right] dt  \label{eq_psi 2agents}
\end{equation}
for each interval simple game $v$ and $i,j\in\{1,2\}$, where $\mathbf{t}_i=t\cdot \mathbf{1}_i$. Then, $\Phi^{a}$ satisfies (E), (NP), and (HIS). 
\end{lemma}

The axioms (E) and (NP) can be checked directly. For (HIS) we mention the corresponding 
constants 
$\left( \lambda _{\Phi ^{a}}(\left\{ 1\right\} ),\gamma _{\Phi ^{a}}(\left\{
1\right\} )\right) =\left( a_{2},a_{2}\right)$ 
and $\left( \lambda
_{\Phi ^{a}}(\left\{ 2\right\} ),\gamma _{\Phi ^{a}}(\left\{ 2\right\}
)\right) =\left( a_{1},a_{1}\right)$ \text{.}
We remark that all power indices satisfying (E), (N), and (HIS) for semi-regular step functions 
can indeed be parameterized as in Lemma~\ref{lemma_hsi_constants 2agents}. 

For $n=2$ including the axiom (S) is sufficient for our claim, but the 
symmetry axiom may also be replaced by some technically weaker axiom. In the context of simple games this reflects 
the fact that (E), (NP), (SymGL) do not characterize the Shapley-Shubik index. However, note that (HIS) is a stronger
requirement than (SymGL) on simple games since the later does
not include the disposition that the constants $\left( \lambda _{\Phi
}\left( S\right) ,\gamma _{\Phi }\left( S\right) \right) $ should depend
only on $\Phi $ and $S$, but not on the game where the improvement occurs.
Moreover, with at least three players, the corresponding of (HIS) axiom on 
simple games characterizes the Shapley-Shubik index when combined with (E) and (NP).

For later usage in the proof of Theorem~\ref{thm_characterization_step_function} we extract the following technical 
result from the proof of Lemma~\ref{lemma_hsi_constants}:
\begin{corollary}
  \label{cor_hsi_constants}
  If $\Phi $ is a power index for interval simple
  games that are semi-regular step functions that simultaneously satisfies
  (E), (NP), and (HIS) and $n\geq 3$, then $\Phi(\tilde{0})=\mathbf{1}/n$, where 
  $\tilde{0}(x)=1$ if $x=\mathbf{1}$ and $\tilde{0}(x)=0$ otherwise. 
\end{corollary}

\medskip

Next we want to prove that (E), (NP), and (HIS) uniquely characterize $\Psi$ within the class of interval simple games that are semi-regular 
step functions. To this end we will show how to obtain any regular step function by a sequence of local increments starting from the 
zero function. It will be necessary to also build up our discretization $\alpha$ step by step.

\begin{definition}
  For a given interval simple game $v$ that is a regular step function with discretization $\alpha\in\mathcal{D}_p$ for 
  some integer $p>1$ and another discretization $\alpha'=\left(0,\alpha_{i_1},\dots,\alpha_{i_{p'-1}},1\right)$, 
  with $i_1<\dots<i_{p'-1}$ and $1\le p'<p$, the coarsened interval simple game is given by
  $$
    v_{\alpha'}(x)=\min \left\{v(y)\,:\,y\in (\alpha)_e, e\in\overline{\mathcal{A}}_p, i_{e'_j-1/2} < e_j< i_{e'_j+1/2}\right\}
  $$
  for all $x\in (\alpha')_{e'}$, where $e'\in \overline{\mathcal{A}}_{p'}$, $i_0=0$, and $i_{p'}=p$. If $x\in[0,1]^n$ is 
  contained on the boundary of some boxes with respect to discretization $\alpha'$, then $v_{\alpha'}(x)$ is uniquely 
  determined by Equation~(\ref{eq_average}).  
\end{definition}

We can easily check that $v_{\alpha'}$ is indeed an interval simple game that is a regular step function. We will apply this coarsening 
only for the cases where $\alpha'=\left(0,\alpha_1,\dots,\alpha_{l-1},1\right)$ for some integer $1\le l\le p$. Examples are given in the appendix.

\begin{theorem}
  \label{thm_characterization_step_function}
  Let $\Phi$ be a power index satisfying (E), (NP), and (HIS) for all interval simple games with $n\ge 3$ players that are semi-regular step functions. Then, we have 
  $\Phi(v)=\Psi(v)$ for every interval simple game $v$ that is a regular step function.
\end{theorem}
\begin{proof}
  Let us build up the regular step function $v$ step by step assuming that $v$ has a discretization $\alpha\in\mathcal{D}_p$ for 
  some integer $p\ge 1$. For each integer $1\le l\le p$ we 
  consider the discretization $\alpha^l=\left(0,\alpha_1,\dots, \alpha_{l-1},1\right)$. For each such $l$ we consider a sequence 
  of local increments, which we call a phase. At the end of each phase our current interval simple game $u$ is the coarsening of 
  $v$ with respect to $\alpha^l$, i.e., for $l=p$ the final interval simple game $u$ coincides with $v$.
  
  For a given $e\in\overline{\mathcal{A}}_p$ let $l\in\mathbb{N}$ be such that $e_j<l$ for all $1\le j\le n$, i.e., 
  $e\in\overline{\mathcal{A}}_l$. By $c_e$ we denote the center of $[\alpha]_e$. 
  Note that $[\alpha]_e\subseteq [\alpha^l]_e$ and $[\alpha]_e$ is indeed the box, with respect to discretization $\alpha$, 
  contained in  $[\alpha^l]_e$ that has the smallest possible coordinates.   
  
  We start from the zero interval simple game $\tilde{0}$ defined by $\tilde{0}(x)=1$ if $x=\mathbf{1}$ and $\tilde{0}(x)=0$ otherwise. 
  Due to Corollary~\ref{cor_hsi_constants} we have $\Phi(\tilde{0})=\Psi(\tilde{0})=\mathbf{1}/n$. For each $l\ge 1$ we are looking at the set of all 
  possible $e\in\overline{\mathcal{A}}_l$, where at least coordinate is equal to $l-\tfrac{1}{2}$. For all such vectors $e,e'$ we choose 
  $e$ before $e'$ if $e\ge e'$. For each such $e$ we modify our current interval simple game $u$ to $u'$, where both are regular step 
  functions. Here we set $u'(x)=u(x)+\varepsilon$ for all $x\in (\alpha^l)_e$ and $u'(x)=u(x)$ for all $x\in[0,1]^n\backslash [\alpha^l]_e$, 
  where $\varepsilon=v(c_e)-u(c_e)$. The remaining values $u'(x)$ are uniquely determined by the property of a regular step function, i.e., 
  by Equation~(\ref{eq_average}). Applying the subsequent Lemma~\ref{lemma_increment} gives $\Phi(u')=\Psi(u')$.

  The statement that at the end of each phase the current interval simple game $u$ coincides with the coarsening $v_{\alpha^l}$ can 
  inductively be concluded from the monotonicity of $v$. Moreover, in all intermediate steps $u$ is always an interval simple game that 
  is a regular step function with respect to discretization $\alpha^l$. 
\end{proof}

\subsubsection{Uniform changes on subsets of the domain}
\label{subsec_uniform_changes}

What is the effect for a given power index $\Phi$ if we change the output values $v(x)$ of an interval simple game uniformly on some full dimensional box 
for a given discretization? It will turn out, that the change in $\Phi(v)$ is uniquely determined if $\Phi$ satisfies (HIS), see Lemma~\ref{lemma_increment}, 
where we consider regular step functions for technical reasons. This is indeed the outsourced part in the proof of Theorem~\ref{thm_characterization_step_function}.   
Definition~\ref{def_his} cannot be applied directly, but has to applied separately for all faces of the given full dimension box.

Assume that $f\colon[0,1]^n\to[0,1]$ is a step function with discretization $\alpha\in\mathcal{D}_p$ for some integer $p\ge 1$. Let further 
$\varepsilon\in\mathbb{R}$ be a constant and $\bar{e}\in \overline{\mathcal{A}}_p$ be the vector that describes a full dimensional box 
according to the given discretization. Note that the entries of $\bar{e}$ satisfy $\bar{e}_i-\tfrac{1}{2}\in\mathbb{N}$, i.e., they are fractional. In order 
to describe the faces of the box $[\alpha]_{\bar{e}}$ we set $\mathcal{F}(\bar{e})=\left\{e\in\mathcal{A}_p\,:\,\bar{e}_i-\tfrac{1}{2}\le e_i\le \bar{e}_i+\tfrac{1}{2}\right\}$. 
For $e\in \mathcal{F}(\bar{e})$ we set $L(e,\bar{e})=\left\{1\le i\le n\,:\,e_i=\bar{e}_i-\tfrac{1}{2}\right\}$ and $U(e,\bar{e})=\left\{1\le i\le n\,:\,e_i=\bar{e}_i+\tfrac{1}{2}\right\}$. 
In words this means that the face $[\alpha]_e$ of $[\alpha]{\bar{e}}$ is located on the lower boundary with respect to coordinate direction $i$ if $i\in L(e,\bar{e})$ and 
on the upper boundary if $i\in U(e,\bar{e})$. In the cases where $i\in I(e,\bar{e}):=N\backslash (L(e,\bar{e})\cup U(e,\bar{e}))$ the $i$th coordinate $x_i$ of $x\in[\alpha]_e$ 
can attain a continuum of values. If for $e\in\mathcal{A}_p$ another step function $g$ is given by $g(x)=f(x)+\varepsilon$ for all $x\in(\alpha)_{e}$ and 
$g(x)=f(x)$ otherwise, then we write $f\localincstep{(\alpha)_{e}}{\varepsilon}g$ or simply $f\localincstep{e}{\varepsilon}g$ whenever the underlying 
discretization is clear from the context. Each such transformation is a local increment $\localinc{S}{\tilde{\varepsilon}}{D}$ if $S$, $\tilde{\varepsilon}$, and $D$ 
are chosen accordingly. To that end, we refine $L(e,\bar{e})$ and $U(e,\bar{e})$ to those indices that are lie on the boundary of $[0,1]^n$, i.e., 
$\overline{L}(e,\bar{e}):=\left\{i\in L(e,\bar{e})\,:\,e_i=0\right\}=\left\{i\in N\,:\,e_i=0\right\}$ and 
$\overline{U}(e,\bar{e}):=\left\{i\in U(e,\bar{e})\,:\,e_i=p\right\}=\left\{i\in N\,:\,e_i=p\right\}$.

\begin{lemma}
  \label{lemma_to_local_increment}
  Let $f\colon[0,1]^n\to[0,1]$ be a step function with discretization $\alpha\in\mathcal{D}_p$, $\bar{e}\in\overline{\mathcal{A}}_p$, 
  $e\in\mathcal{F}(\bar{e})$, $\varepsilon\in\mathbb{R}$, $f\localincstep{e}{\varepsilon}g$, and $D=\left[\alpha\right]_{e_{N\backslash S}}:=
  \underset{j\in N\backslash S}{\bigtimes}\left[\alpha_{e_j-1/2},\alpha_{e_j+1/2}\right]$.
  \begin{enumerate}
    \item[(1)] If $\overline{L}(e,\bar{e})\neq \emptyset\,\wedge \overline{U}(e,\bar{e})\neq \emptyset$ or $\overline{L}(e,\bar{e})=\emptyset\wedge\overline{U}(e,\bar{e})= \emptyset$, 
               then $f\localinc{S}{\tilde{\varepsilon}}{D}g$ for arbitrary $S\in 2^N$ and $\tilde{\varepsilon}=0$.
    \item[(2)] If $\overline{U}(e,\bar{e})\neq\emptyset$ and $\overline{L}(e,\bar{e})=\emptyset$, 
               then $f\localinc{S}{\tilde{\varepsilon}}{D}g$ for $S=\overline{U}(e,\bar{e})$ and $\tilde{\varepsilon}=\varepsilon$.           
    \item[(3)] If $\overline{L}(e,\bar{e})\neq\emptyset$ and $\overline{U}(e,\bar{e})=\emptyset$, 
               then $f\localinc{S}{\tilde{\varepsilon}}{D}g$ for $S=\overline{L}(e,\bar{e})$ and $\tilde{\varepsilon}=-\varepsilon$.           
  \end{enumerate}
\end{lemma}
\begin{proof}
   First note that $\overline{L}(e,\bar{e})\cap \overline{U}(e,\bar{e})= \emptyset$, so that the case analysis is exhaustive. For $T\in \{\emptyset,N\}$ we have 
   $\Delta(f,x_{-T})=\Delta(g,x_{-T})$ for all $x_{-T}\in(0,1)^{N\backslash T}$ in general. 
   
   Now let $T\in 2^N\backslash\{\emptyset,N\}$. If exists some $i\in T$ but 
   $i\notin \overline{L}(e,\bar{e})\cup \overline{U}(e,\bar{e})$, then $g(\mathbf{1}_T,x_{-T})=f(\mathbf{1}_T,x_{-T})$ and $g(\mathbf{0}_T,x_{-T})=f(\mathbf{0}_T,x_{-T})$ 
   for all $x_{-T}\in (0,1)^{N\backslash T}$. Similarly, if there exist some $i,j\in T$ with $i\in \overline{L}(e,\bar{e})$ and $j\in \overline{U}(e,\bar{e})$, then 
   $g(\mathbf{1}_T,x_{-T})=f(\mathbf{1}_T,x_{-T})$ and $g(\mathbf{0}_T,x_{-T})=f(\mathbf{0}_T,x_{-T})$ for all $x\in (0,1)^{N\backslash T}$. Thus, (1) is true 
   and it remains to check the cases $T=\overline{L}(e,\bar{e})\neq \emptyset$ or $T=\overline{U}(e,\bar{e})\neq \emptyset$.
   
   In the following we assume that exactly one of the sets $\overline{L}(e,\bar{e})$ and $\overline{U}(e,\bar{e})$ is non-empty. If 
   $T=\overline{U}(e,\bar{e})\neq \emptyset$, then $g(\mathbf{1}_T,x_{-T})=f(\mathbf{1}_T,x_{-T})+\varepsilon$ if $x_{-T}\in(\alpha)_{e_{N\backslash T}}$ and 
   $g(\mathbf{1}_T,x_{-T})=f(\mathbf{1}_T,x_{-T})$ otherwise. 
   For all $x_{-T}\in (0,1)^{N\backslash T}$ we have $g(\mathbf{0}_T,x_{-T})=f(\mathbf{0}_T,x_{-T})$. If $T=\overline{L}(e,\bar{e})\neq \emptyset$, then 
   $g(\mathbf{1}_T,x_{-T})=f(\mathbf{1}_T,x_{-T})$ for all $x_{-T}\in(0,1)^{N\backslash T}$ and 
   $g(\mathbf{0}_T,x_{-T})=f(\mathbf{0}_T,x_{-T})+\varepsilon$ for all $x_{-T}\in(\alpha)_{e_{N\backslash T}}$ and $g(\mathbf{0}_T,x_{-T})=f(\mathbf{0}_T,x_{-T})$ 
   otherwise. This gives (2) and (3), respectively. 
\end{proof}   

\begin{corollary}
  \label{cor_matter}
  Let $\Phi$ be a power index satisfying (HIS) for all interval simple games that are semi-regular step functions, $f\colon[0,1]^n\to[0,1]$ be an 
  interval simple game that is a semi-regular step function with discretization $\alpha\in\mathcal{D}_p$, $\bar{e}\in\overline{\mathcal{A}}_p$, 
  $e\in\mathcal{F}(\bar{e})$, $\varepsilon\in\mathbb{R}_{>0}$, and $f\localincstep{e}{\varepsilon}g$. We have $\Phi(f)\neq\Phi(g)$ iff 
  $L(e,\bar{e})=\overline{L}(e,\bar{e})$, $U(e,\bar{e})=\overline{U}(e,\bar{e})$, and $\overline{L}(e,\bar{e})\cup \overline{U}(e,\bar{e})\neq\emptyset$.   
\end{corollary}

In other words, a transformation $f\localincstep{e}{\varepsilon}g$ has an effect for a power index $\Phi$, satisfying (HIS) if and only if the 
non-empty set of coordinates that is fixed in $[\alpha]_e$ lies on the boundary of $[0,1]^n$ and $\varepsilon\neq 0$.

\begin{lemma}
  \label{lemma_increment}
  Let $\Phi$ be a power index with $n\ge 3$ satisfying (HIS), (E), and (NP) for all interval simple games that are semi-regular step functions and 
  $u$ and $u'$ be interval simple games that are regular step functions with discretization $\alpha\in\mathcal{D}_p$ that satisfy $\Phi(u)=\Psi(u)$. 
  If $u'(x)=u(x)$ for all $x\in[0,1]^n\backslash [\alpha]_{\bar{e}}$ and $u'(x)=u(x)+\varepsilon$ for all $x\in(\alpha)_{\bar{e}}$ for some $\varepsilon\in\mathbb{R}_{\ge 0}$ 
  and $\bar{e}\in\overline{\mathcal{A}}_p$, then $\Phi(u')=\Psi(u')$.
\end{lemma}
\begin{proof}
  In order to transform $u$ into $u'$ we apply step of the form $\localincstep{e}{\varepsilon/|E(c_e)|}$ for all faces $e\in\mathcal{F}(\bar{e})$ of $[\alpha]_{\bar{e}}$, 
  where the normalization $\varepsilon/|E(c_e)|$ guarantees that the resulting step function $u'$ is indeed regular.
  
  We loop over all possible faces $e\in\mathcal{F}(\bar{e})$, where we choose $e$ before $e'$ if $e\ge e'$. This ordering guarantees 
  that $\localincstep{e}{\varepsilon/|E(c_e)|}$ turns an interval simple game into an interval simple game, i.e., monotonicity is respected. Moreover, we stay 
  in the class of semi-regular step functions. Thus, given (HIS), (E), and (NP), we can apply Lemma~\ref{lemma_hsi_constants} and Lemma~\ref{lemma_to_local_increment} to conclude   
  $\Phi(\tilde{u})=\Psi(\tilde{u})$ for all intermediate semi-regular step functions $\tilde{u}$.
\end{proof}

Note that in the case where either $\bar{e}=\tfrac{1}{2}\cdot \mathbf{1}$ or $\bar{e}=\left(p-\tfrac{1}{2}\right)\cdot \mathbf{1}$ we even have $\Phi(u)=\Phi(u')$ 
and $\Psi(u)=\Psi(u')$. In general it is quite hard to explicitly quantify the effect of a transformation as described in Lemma~\ref{lemma_increment}. For some 
integer $l\ge 2$ we consider the discretization $\alpha=\tfrac{1}{l}\cdot(0,1,\dots,l)$ as an example. Let $n=3$ and $\bar{e}=\left(l-\tfrac{1}{2},\tfrac{1}{2},l-\tfrac{1}{2}\right)$. 
In Table~\ref{table_effect_uniform_increase} we have listed the changes of the corresponding power distribution $\Phi$ for all faces of $[\alpha]_{\bar e}$ that matter, 
see Corollary~\ref{cor_matter}. 

\begin{table}[htp]
\begin{center}
  \begin{tabular}{llllll}
  \hline
  face $e$ of $\bar{e}$ & $S$ & $\operatorname{vol}(D)$ & change $\Phi_1$ & change $\Phi_2$ & change $\Phi_3$\\
  \hline
  $x_1=1$, $x_3=1$ & $\{1,3\}$ & $l^{-1}$   & $+\frac{\varepsilon}{6l}$ & $-\frac{\varepsilon}{3l}$ & $+\frac{\varepsilon}{6l}$\\
  $x_1=1$          & $\{1\}$   & $l^{-2}$ & $+\frac{\varepsilon}{3l^2}$ & $-\frac{\varepsilon}{6l^2}$ & $-\frac{\varepsilon}{6l^2}$\\
  $x_3=1$          & $\{3\}$   & $l^{-2}$ & $-\frac{\varepsilon}{6l^2}$ & $-\frac{\varepsilon}{6l^2}$ & $+\frac{\varepsilon}{3l^2}$\\
  $x_2=0$          & $\{2\}$   & $l^{-2}$ & $+\frac{\varepsilon}{6l^2}$ & $-\frac{\varepsilon}{3l^2}$ & $+\frac{\varepsilon}{6l^2}$\\
  \hline
  \end{tabular}
  \caption{Effect of a uniform increase on some full-dimensional subcube for $\Phi$.}
  \label{table_effect_uniform_increase}
\end{center}
\end{table}

For the special case $\varepsilon=1$ and $l=2$ the players 1 and 3 increase their power by $\tfrac{1}{6}$, while player~2 reduces its power 
by $\tfrac{1}{3}$. More generally, for $\bar{e}=\left(\tfrac{1}{2}\cdot\mathbf{1}_L,\left(l-\tfrac{1}{2}\right)\cdot\mathbf{1}_U\right)$ for some disjoint sets
$L,U\neq \emptyset$ with $L\cup U=N$, the total change for $\Phi_i$ is given by
\begin{eqnarray}   
  &&\sum_{\{i\}\subseteq T\subseteq L} -\frac{\varepsilon}{l^{n-t}}\cdot \frac{(t-1)!(n-t)!}{n!}
  +\sum_{\emptyset\subsetneq T\subseteq L\backslash\{i\}} \frac{\varepsilon}{l^{n-t}}\cdot \frac{t!(n-t-1)!}{n!}\nonumber\\ 
  &&+\sum_{\{i\}\subseteq T\subseteq U} \frac{\varepsilon}{l^{n-t}}\cdot \frac{(t-1)!(n-t)!}{n!}
  +\sum_{\emptyset\subsetneq T\subseteq U\backslash\{i\}} -\frac{\varepsilon}{l^{n-t}}\cdot \frac{t!(n-t-1)!}{n!}\label{eq_corner_increase},
\end{eqnarray}
where $t=|T|$. For $\varepsilon=1$ and $l=2$ this complicated expression simplifies to an overall increase of 
$\frac{(|U|-1)!\cdot(n-|U|)!}{n!}$ for every player $i\in U$ and a decrease of $\frac{|U|!\cdot(n-|U|-1)!}{n!}$ for all other players. 
Instead of a combinatorial proof we refer to the power distribution of a $(2,2)$-simple game in Proposition~\ref{prop_embed} and note that we 
actually obtain Equation~(\ref{eq_SSI_add_S}) for a simple game (after adding one winning coalition). For general $(j,k)$ simple games 
and their natural embedding as an interval simple game the situation is more complicated. Here we might have $L\cup U\neq N$, i.e., the full-dimensional 
boxes of the discretization do not need to meet the boundary of $[0,1]^n$ in all coordinate directions. Moreover, we are not aware of a generalization 
of Equation~(\ref{eq_SSI_add_S}) to $(j,k)$ simple games.
 
Finally, we mention that the technical normalization in $\localincstep{e}{\varepsilon/|E(c_e)|}$, i.e., the division by $|E(c_e)|$ does not have 
an effect on the resulting power distribution. I.e., due to Corollary~\ref{cor_matter} we have an effect only if $|E(c_e)|=1$. 
 
\subsection{A characterization of $\Psi$}
\label{subsec_characterization}

In Theorem~\ref{thm_characterization_step_function} we have already characterized $\Psi$ by three axioms on the subset of (regular) step functions 
of interval simple games. The idea now is to consider for a given approximable interval simple game $v$ an approximation sequence $(f^h)_{h\in \mathbb{N}}$ of 
regular step functions. Then, $f^h$ tends to $v$ and $\Psi(v)$ should be given by $\lim_{h\to\infty} \Psi(f^h)$.

\begin{definition}
  \label{definition_limit_psi} A power index $\Phi$ is called \emph{discretizable} (D) if for any approximable interval simple game 
  $v$ and any approximation sequence $\left(f^h\right)_{h\in\mathbb{N}}$ the limit $\lim_{h\to\infty} \Phi\left(f^h\right)$ exists and coincides 
  with $\Phi(v)$.
\end{definition}  

\begin{proposition}
  \label{prop_Psi_discretizable}
  The power index $\Psi$ is discretizable.
\end{proposition}
\begin{proof}
  As argued in Section~\ref{sec_the_index} the integrals in the definition of $\Psi$ even exist applying the concept of the 
  Riemann integral due to the monotonicity of interval simple games. Since the Riemann integral itself is defined via 
  limits over sequences of step functions the statement naturally follows.
\end{proof}

For approximation sequences we may also restrict ourselves to uniform discretizations. The power index $\Psi$ has the limit property of Definition~\ref{definition_limit_psi} 
also for more general sequences of (regular) step functions. However, it makes sense to assume a (seemingly) weaker condition for axiom (D) if 
we are able to conclude uniqueness of the power index in the end.

Before we state the final axiomatization we want to state two further examples to highlight the technical subtleties of power indices 
for general interval simple games.  

Let $n\ge 1$ be an integer and $v$ be the interval simple game for $n$ players with $v(\mathbf{0})=0$ and $v(x)=1$ for all 
$x\in[0,1]^n\backslash\{\mathbf{0}\}$. For every $\emptyset\subseteq T\subsetneq N$ we have $C(v,T)=0$ and $C(v,N)=1$, so that 
$\Psi_i(v)=\tfrac{1}{n}$ for all $i\in N$. For every integer $p\ge 1$ and every discretization $\alpha\in\mathcal{D}_p$ there 
exists a regular step function $f$ with discretization $\alpha$ and $f(c_e)=1$ for all $e\in \overline{\mathcal{A}}_p$, where 
$c_e$ denotes the center of $[\alpha]_e$. For all these step functions we also have $\Psi(f)_i=\tfrac{1}{n}$ for all $i\in N$. However this 
is not the only choice for the step functions. For each $e\in\overline{\mathcal{A}}_p$, where $e_h=\tfrac{1}{2}$ for some index $h\in N$, we may 
set $f(c_e)=0$ instead of $f(c_e)=1$, since $\mathbf{0}\in[\alpha]_e$. Even if we ignore the symmetry assumptions for semi-regular 
step functions, we can see that those changes of values of $f$ at the {\lq\lq}boundary{\rq\rq} of $[0,1]^n$ have an impact on $\Psi(f)$ 
that tends to zero if $\omega(\alpha)$ tends to zero.

We remark that we can embed every $(j,k)$ simple game as an interval simple game, but for the other direction we need further conditions, 
i.e., $f(x)=0$ for all $x\in (\alpha)_{\frac{1}{2}\cdot\mathbf{1}}$ and $f(x)=1$ for all $x\in (\alpha)_{\left(p-\frac{1}{2}\right)\cdot\mathbf{1}}$. 

\begin{theorem}
  \label{thm_characterization}
  Let $\Phi$ be a power index that satisfies (E), (NP), (HIS), and (D) for interval simple games with $n\ge 3$ players. Then we have $\Phi(v)=\Psi(v)$ for every 
  interval simple game that is approximable, i.e., $\Phi=\Psi$.
\end{theorem}
\begin{proof}
  Combine Proposition~\ref{prop_Psi_discretizable} with Theorem~\ref{thm_characterization_step_function} using axiom (D).
\end{proof}

We remark that the four stated axioms for $\Psi$ are independent from each other:
\begin{itemize}
  \item The power index $2\cdot \Psi$ satisfies (NP), (HIS), and (D), but not (E).
  \item Denote by $\operatorname{ED}$ the \emph{equal division} power index which assigns $\tfrac{1}{n}$ to
        each player for ever interval simple game $v$. Then the power index $\tfrac{1}{2}\Psi+\tfrac{1}{2}\operatorname{ED}$ 
        satisfies (E), (HIS) and (D), but not (NP).
 \item In Proposition~\ref{prop_generalized_measure_point} and Proposition~\ref{prop_generalized_measure_density} we have presented 
       parametric classes of power indices for interval simple games that satisfy (E), (A), (S), (NP), and (TP). Those from 
       Proposition~\ref{prop_generalized_measure_density} also satisfy (D). In general they do not satisfy (HIS) as shown by 
       an example in Section~\ref{sec_the_index}. Another concrete example is given by the mapping $v\mapsto \Psi(v^2)$, 
       as already mentioned in Section~\ref{sec_the_index}, which satisfies (E), (NP), and (D), but not (HIS).
 \item The construction of a power index that satisfies (E), (NP), and (HIS) but not (D) is a bit more technically involved. On the set 
       of interval simple games with $n$ players we can define an equivalence relation, where two interval simple games are in the 
       same class if one of them can be obtained from the other by a finite sequence of local improvements. It can be shown 
       that there exists more than one equivalence class if $n\ge 2$. For e.g.\ $v(x)=\prod_{i=1}^n x_i^i$ we can define 
       $\Phi_i(v)=\frac{1}{n(n+1)}\cdot (1,2,\dots,n)\neq \Psi(v)$, see Proposition~\ref{prop_formula_exponential_product}. 
       For every interval simple game $u$ within the same equivalence from $v$ the value of $\Phi(u)$ is defined by $\Phi(v)$ via 
       (HIS). Moreover, we cannot lose efficiency by those transformation and no player can turn into a null player. For 
       every interval simple game $u'$ that is contained in a different equivalence class than $v$, we set $\Phi(u')=\Psi(u')$, 
       so that $\Phi$ satisfies (E), (NP), and (HIS), but not (D).                 
\end{itemize}

We remark that (HIS) is a quite strong axiom. The situation is similar for simple games and (SymGL), which e.g.\ 
implies the transfer axiom, see \cite{laruelle2001shapley}. 

\section{Conclusion and open problems}
\label{sec_conclusion}
The Shapley-Shubik index $\operatorname{SSI}$ was designed to evaluate the power distribution in committee systems 
drawing binary decisions and is one of the most established power indices. It was generalized in the literature to 
decisions with more than two levels of approval in the input and output, i.e., so-called $(j,k)$ simple games. A 
fitting axiomatization is still the topic of current research, see \cite{shapleybookjosep}. Special cases for $(j,k)$ 
are e.g.\ treated in \cite{bernardifreixas2018new}. If we consider the limit of $j$ and $k$ to infinity, then we end 
up with a continuum of options. Those games, that we call interval simple games here, were e.g.\ introduced in 
\cite{kurz2014measuring,kurz2018importance}. In the same papers a generalization of the Shapley-Shubik index, i.e., $\Psi$, 
was proposed and motivated by a generalized version of the roll call interpretation of the Shapley-Shubik index. 
Here we prove the first axiomatization for $\Psi$ and show that the Shapley-Shubik index for simple games, as well as 
for $(j,k)$~simple games, occurs as a special discretization. This relation and the closeness of the stated axiomatization 
to the classical case suggests to speak of the Shapley-Shubik index for games with interval decisions that can also be generalized 
to a value. In that context we have shown that generalized versions of the classical axioms for simple games, see 
Proposition~\ref{prop_generalized_measure_point}, are not sufficient any more.

For the newly introduced axiom (HIS) we gave some justification and remarked its similarity to the axiom (SymGL) for simple games. However, 
as shown in e.g.\ (\ref{eq_corner_increase}), the implications of (HIS) are much more farreaching than the implications of  (SymGL), which 
is a more direct axiom tailored for simple games. Again, we mention that we are not aware of a generalization of Equation~(\ref{eq_SSI_add_S}) 
to $(j,k)$ simple games and pose this as an open problem.

Our emphasis on step functions is not really essential for our approach. One main motivation is the proposed natural embedding of $(j,k)$ simple 
games as interval simple games. Another motivation was that we do not want to dive too much into the mathematical details of integrability and 
approximability. Especially the later topic should be worth a more mathematical treatment. The introduction of regular step functions was necessary 
to precisely describe the above mentioned natural embedding. Nonetheless $\Psi$ can be defined as an $n$-dimensional integral, the corresponding 
power vector is not uniquely defined for an ordinary step functions. The precise values on lower dimensional faces can be essential. Actually, one 
can be more precise, see Corollary~\ref{cor_matter}. I.e., faces that touch the boundary of $[0,1]^n$ are the essential ones. Stated more directly, 
the values of an interval simple game $v$ in the interior of its domain are more or less irrelevant for $\Psi$. This property might be analyzed and 
criticized from a more general and non-technical point of view. However, our rigorous technical analysis uncovers this fact for the first time, while 
it is also valid for the $\operatorname{SSI}$ variant for $(j,k)$ simple games. I.e., in e.g.\ a $(4,4)$ simple game $v$ for $n=3$ voters the value of 
$v(1,2,1)$ can be changed to $0$, $1$, $2$, or $3$ without any direct effect for the power distribution of the players. Of course monotonicity implies 
some possible indirect changes of other function values, which then can have an effect for the power distribution. For simple games there are no 
{\lq\lq}internal{\rq\rq} vote vectors. In any case this {\lq\lq}boundary dependence{\rq\rq} should be studied and interpreted in more detail. Semi-regular 
step function are just a technical artefact in order to allow inductive arguments using (HIS), while being able to state all technical details and 
subtleties without glossing over.

Nevertheless the power measurement $\Psi$ can be criticized, we think that it is an reasonable object that is worth to be studied in more detail. 
More explicit formulas for special parametric classes of interval simple games, c.f.\ Section~\ref{sec_the_index}, might be a promising direction 
for further research. Of course, also other axiomatizations of the Shapley-Shubik index or the Shapley value should be studied, whether they can be 
generalized to interval simple games. In any case, we think that interval simple games with their natural embedding to $(j,k)$ simple games 
are a good yardstick to check how specific some approaches in the current and in the future literature are when considering specific instances 
of $(j,k)$ simple games.      

\section*{Acknowledgment}
Hilaire Touyem benefits from a financial support of the CETIC (Centre d'Excellence Africain en
Technologies de l'Information et de la Communication) Project of the University of
Yaounde~I.


\appendix

\section{A working example}

Let $\alpha =\left( 0,\tfrac{1}{4},\tfrac{1}{2},1\right) $, $n=2$ and
consider the regular step function $v$ given by

\begin{equation*}
\begin{tabular}{|c|c|c|c|c|c|c|c|c|c|}
\hline
&  &  &  &  &  &  &  &  &  \\
$e$ & $\left( \tfrac{1}{2},\tfrac{1}{2}\right) $ & $\left( \tfrac{1}{2},\tfrac{3}{2}\right) $ & $\left( \tfrac{3}{2},\tfrac{1}{2}\right) $ & $%
\left( \tfrac{1}{2},\tfrac{5}{2}\right) $ & $\left( \tfrac{5}{2},\tfrac{1}{2}\right) $ & $\left( \tfrac{3}{2},\tfrac{3}{2}\right) $ & $\left(
\tfrac{3}{2},\tfrac{5}{2}\right) $ & $\left( \tfrac{5}{2},\tfrac{3}{2}\right) $ & $\left( \tfrac{5}{2},\tfrac{5}{2}\right) $ \\
&  &  &  &  &  &  &  &  &  \\ \hline
&  &  &  &  &  &  &  &  &  \\
$c_{e}$ & $\left( \tfrac{1}{8},\tfrac{1}{8}\right) $ & $\left( \tfrac{1}{8},%
\tfrac{3}{8}\right) $ & $\left( \tfrac{3}{8},\tfrac{1}{8}\right) $ & $\left(
\tfrac{1}{8},\tfrac{3}{4}\right) $ & $\left( \tfrac{3}{4},\tfrac{1}{8}%
\right) $ & $\left( \tfrac{3}{8},\tfrac{3}{8}\right) $ & $\left( \tfrac{3}{8}%
,\tfrac{3}{4}\right) $ & $\left( \tfrac{3}{4},\tfrac{3}{8}\right) $ & $%
\left( \tfrac{3}{4},\tfrac{3}{4}\right) $ \\
&  &  &  &  &  &  &  &  &  \\ \hline
&  &  &  &  &  &  &  &  &  \\
$v\left( c_{e}\right) $ & $0.1$ & $0.2$ & $0.3$ & $0.4$ & $0.5$ & $0.6$ & $%
0.7$ & $0.8$ & $0.9$ \\
&  &  &  &  &  &  &  &  &  \\ \hline
\end{tabular}%
\end{equation*}


In the proof of Theorem \ref{thm_characterization_step_function}, it is
shown how one can move from the zero interval simple game to any other
interval simple game that is regular. Below we give an implementation of
this construction, for illustration, to determine $\Psi (v)$ for the game $v$
just presented. In the intermediate steps, only semi-regular step functions are used. Thus, 
we only provide the values of the game at the center of each box, or on an
edge where the value of game differs from the one at the center; any other
edge is completed using Equation~(\ref{eq_average}).

We do not enumerated all tiny steps of the form $f\localincstep{e}{\varepsilon}g$ explicitly. 
Due to Lemma~\ref{lemma_to_local_increment} we only consider steps of the form $\localinc{S}{\varepsilon}{D}$ 
for all $\emptyset\neq S\subseteq \overline{U}(e,\bar{e})$ and $\localinc{T}{-\varepsilon}{D}$ 
for all $\emptyset\neq T\subseteq \overline{L}(e,\bar{e})$. Whenever we have two such steps (for $n=2$ 
we cannot have more), we increase the function values of the interior of the corresponding rectangle 
at the second step.

Going through all the details of the inductive construction is rather tedious, but might 
help the reader to comprehend the technical details of our definitions and lemmas.   

\begin{figure}[htp]
\centering
\includegraphics[width=7cm]{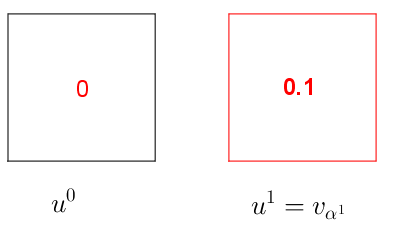}
\caption{Moves at phase 1}
\label{phase_1}
\end{figure}

\subsubsection*{\textbf{Phase 1}}

We start with the interval simple game given by $u^{0}(x)=0$ for all $x\in
[0,1]^{2}\backslash \left\{ \left( 1,1\right) \right\} $, where we
clearly have $\Psi (u^{0})=\left( \tfrac{1}{2},\tfrac{1}{2}\right) $. Next
we consider the local increment $u^{0}\localinc{\emptyset}{0.1}{[0,1]^{2}}u^{1}$, 
so that $\Psi (u^{1})=\left( \tfrac{1}{2},\tfrac{1}{2}\right)$; see Figure~\ref{phase_1}.

\begin{figure}[htp]
\centering
\includegraphics[width=10cm]{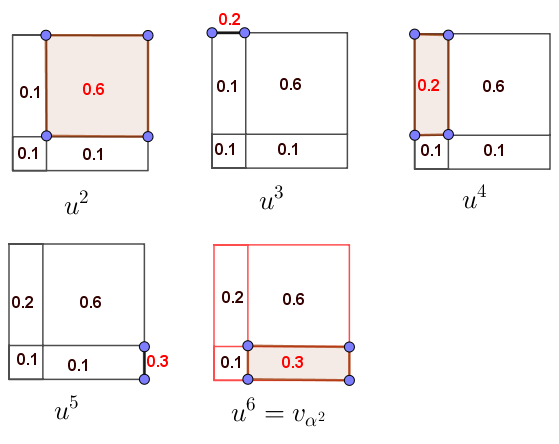}
\caption{Moves at phase 2}
 \label{phase_2}
\end{figure}

\subsubsection*{\textbf{Phase 2}}

At this stage, the current game $u=u^{1}$ already coincides with $v$ on the
box $\left[ \alpha \right] _{\left( \tfrac{1}{2},\tfrac{1}{2}\right) }$. We now consider the
discretization $\alpha ^{2}=\left( 0,\tfrac{1}{4},1\right) $. We have to
order the vectors in $\overline{\mathcal{A}}_{2}\backslash \left\{ \left(\tfrac{1}{2},\tfrac{1}{2}\right)\right\} $
according to $\geq $. There are exactly two possibilities: $\left(\tfrac{3}{2},\tfrac{3}{2}\right)$, $\left(\tfrac{3}{2},\tfrac{1}{2}\right)$,
$\left(\tfrac{1}{2},\tfrac{3}{2}\right)$ and $\left(\tfrac{3}{2},\tfrac{3}{2}\right)$, $\left(\tfrac{1}{2},\tfrac{3}{2}\right)$, 
$\left(\tfrac{3}{2},\tfrac{1}{2}\right)$. Note that we are dealing with a
partial order, i.e., $\left(\tfrac{3}{2},\tfrac{3}{2}\right)>\left(\tfrac{3}{2},\tfrac{1}{2}\right)$ and 
$\left(\tfrac{3}{2},\tfrac{3}{2}\right)>\left(\tfrac{1}{2},\tfrac{3}{2}\right)$, but we cannot compare $\left(\tfrac{3}{2},\tfrac{1}{2}\right)$ 
and $\left(\tfrac{1}{2},\tfrac{3}{2}\right)$. For $e=\left(\tfrac{3}{2},\tfrac{3}{2}\right)$ we compute $S=\{1,2\}$ and $T=\emptyset $,
$D=[\tfrac{1}{4},1]^{2}$, and $\varepsilon =v(c_{\left(\tfrac{3}{2},\tfrac{3}{2}\right)})-v(c_{\left(\tfrac{1}{2},\tfrac{1}{2}\right)})=0.50$.
Setting $u^{1}\localinc{S}{\varepsilon}{D}u^{2}$ we obtain $%
\Psi (u^{2})=\left( \tfrac{1}{2},\tfrac{1}{2}\right) $. Similarly, for $%
e=\left( \tfrac{1}{2},\tfrac{3}{2}\right) $, we have $S=\left\{ 2\right\} $ and $T=\left\{
1\right\} $. First, $u^{2}\localinc{\left\{ 2\right\}}{0.1}{\left[ 0,\frac{1}{4}\right]} 
u^{3}$ gives $\Psi (u^{3})=\left( \tfrac{1}{2}-%
\frac{1}{80},\tfrac{1}{2}+\frac{1}{80}\right) $. Next $u^{3}\localinc{\left\{
1\right\}}{-0.1}{\left[ \frac{1}{4},1\right] }u^{4}$
implies $\Psi (u^{4})=\left( \frac{39}{80}-\frac{3}{80},\frac{41}{80}+\frac{3%
}{80}\right) $. For $e=\left( \tfrac{3}{2},\tfrac{1}{2}\right) $, we have $S=\left\{ 1\right\} $
and $T=\left\{ 2\right\} $. The first move to update the value of the game
in the box $\left[ \alpha \right] _{e}$ first consists in $u^{4}\localinc
{\left\{ 1\right\}}{0.2}{\left[ 0,\frac{1}{4}\right] }u^{5}$
gives $\Psi (u^{5})=\left( \frac{9}{20}+\frac{2}{80},\frac{11}{20}-\frac{2}{%
80}\right)$. Finally, $u^{5}\localinc{\left\{ 2\right\}}{-0.2}{\left[ \frac{1}{4},1\right] }u^{6}$ implies 
$\Psi (u^{6})=\left( \frac{19}{40}+\frac{6}{80},\frac{21}{40}-\frac{6}{80}\right)$. Figure~\ref{phase_2}
shows the entire sequence of moves at this phase.

\begin{figure}[htp]
\centering
\includegraphics[width=12cm]{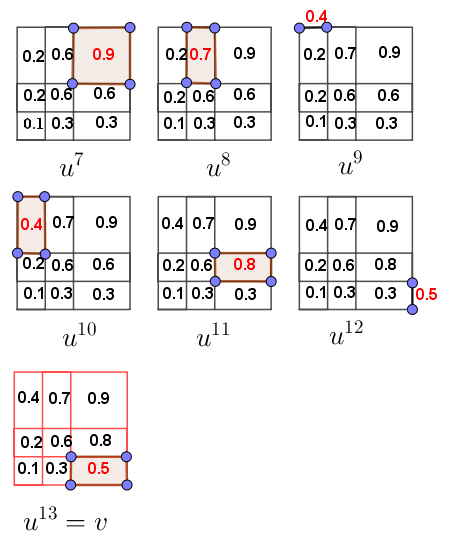}
\caption{Moves at phase 3}
\label{phase_3}
\end{figure}

\subsubsection*{\textbf{Phase 3}}

Note that the current game $u=u^{6}$ already coincides with $v$ on the box $%
\left[ \alpha \right] _{e}$ for $e\in \overline{\mathcal{A}}_{2}$. We consider the
discretization $\alpha =\alpha ^{3}=\left( 0,\tfrac{1}{4},\tfrac{1}{2}%
,1\right) $. We have to order the vectors in $\overline{\mathcal{A}}_{3}\backslash
\overline{\mathcal{A}}_{2}$ according to $\geq $. As above, there are several
possibilities that emerge to the same result. Here we consider $\left(\tfrac{5}{2},\tfrac{5}{2}\right)$, 
$\left(\tfrac{3}{2},\tfrac{5}{2}\right)$, $\left(\tfrac{1}{2},\tfrac{5}{2}\right)$, $\left(\tfrac{5}{2},\tfrac{3}{2}\right)$, 
$\left(\tfrac{5}{2},\tfrac{3}{2}\right)$, and $\left(\tfrac{5}{2},\tfrac{1}{2}\right)$. For $e=\left(\tfrac{5}{2},\tfrac{5}{2}\right)$ 
we have $S=\{1,2\}$ and $T=\emptyset$ and $u^{6}\localinc{N}{0.3}{\left[\tfrac{1}{2},1\right]^{2}}u^{7}$. 
It follows that $\Psi (u^{7})=\left( \frac{11}{20},%
\frac{9}{20}\right) $. Next for $e=\left( \tfrac{3}{2},\tfrac{5}{2}\right) $, we have $S=\left\{
2\right\} $ and $T=\emptyset$. Setting $u^{7}\localinc{\left\{ 2\right\}
}{0.1}{\left[ \frac{1}{4},\frac{1}{2}\right]}u^{8}$ gives $%
\Psi (u^{8})=\left( \frac{11}{20}-\frac{1}{80},\frac{9}{20}+\frac{1}{80}%
\right)$. For $e=\left( \tfrac{1}{2},\tfrac{5}{2}\right) $, we have $S=\left\{ 2\right\} $ and $%
T=\left\{ 1\right\} $. The first move to update the value of the game in $%
\left[ \alpha \right] _{e}$ first consists in $u^{8}\localinc{\left\{
2\right\}}{0.2}{\left[ 0,\frac{1}{4}\right] }u^{9}$ which
gives $\Psi (u^{9})=\left( \frac{43}{80}-\frac{2}{80},\frac{37}{80}+\frac{2}{%
80}\right) $. Next, $u^{9}\localinc{\left\{ 1\right\} }{-0.2}{\left[ \frac{1}{2}%
,1\right] }u^{10}$ yields $\Psi (u^{10})=\left( \frac{41}{%
80}-\frac{2}{40},\frac{39}{80}+\frac{2}{40}\right) $. For $e=\left(\tfrac{5}{2},\tfrac{3}{2}\right) $, we have $S=\left\{ 1\right\} $ and $T=\emptyset $. Setting $%
u^{10}\localinc{\left\{ 1\right\}}{0.2}{\left[ \frac{1}{4},\frac{1}{2}\right] }%
u^{11}$ gives $\Psi (u^{11})=\left( \frac{37}{80}+\frac{2}{%
80},\frac{43}{80}-\frac{2}{80}\right) $. For $e=\left( \tfrac{5}{2},\tfrac{1}{2}\right) $, we have
$S=\left\{ 1\right\} $ and $T=\left\{ 2\right\} $. Two moves are needed to
update the value of the game in $\left[ \alpha \right] _{e}$. First, $u^{11}%
\localinc{\left\{ 1\right\}}{0.2}{\left[ 0,\frac{1}{4}\right] }u^{12}$ gives 
$\Psi (u^{12})=\left( \frac{39}{80}+\frac{2}{%
80},\frac{41}{80}-\frac{2}{80}\right)$. Finally, $u^{12}\localinc{\left\{
2\right\}}{-0.2}{\left[ \frac{1}{2},1\right] }u^{13}$ leads
$\Psi (u^{13})=\left( \frac{41}{80}+\frac{2}{40},\frac{39}{80}-\frac{2}{40}%
\right) =\left( \frac{9}{16},\frac{7}{16}\right) $. All those moves are
illustrated in Figure~\ref{phase_3}.

\end{document}